\documentclass[12pt,titlepage]{article}
\usepackage{amsfonts}
\usepackage{eurosym}
\usepackage{amsmath,amsthm,amssymb}
\usepackage{graphicx}
\usepackage{array}
\usepackage{bbm}
\usepackage{subfig}
\usepackage{float}
\usepackage[mathscr]{eucal}
\usepackage{epsfig, graphics}
\usepackage{amsmath,amssymb}
\usepackage[latin1]{inputenc}
\usepackage{graphicx}
\usepackage{epsfig}
\usepackage{amsfonts}
\usepackage{caption}
\usepackage{enumerate}
\usepackage[latin1]{inputenc}

\def\bkR{{\rm I\kern-.17em R}}

\def \1n{1\hskip -3pt \mbox{N}}

\def \Sum {\displaystyle \sum }

\newtheorem{Proposition}{Proposition}

\newtheorem{Corollary}{Corollary}
\theoremstyle{Definition}
\newtheorem{Remark}{Remark}

\numberwithin{equation}{section}
\numberwithin{table}{section}

\newtheorem{lm}{Lemma}
\input{tcilatex}

\begin{document}

\title{The Causal-Noncausal Tail Processes}

\author{Gouri\'eroux, C.,$^{(1)}$, Y., Lu $^{(2)}$, and C.Y., Robert $^{(3)}$ }

\addtocounter{footnote}{1} \footnotetext{University of Toronto, Toulouse School of Economics and CREST.}
\addtocounter{footnote}{1} \footnotetext{Department of Mathematics and Statistics.  Concordia University. Corresponding author.  Email: yang.lu@concordia.ca}

\addtocounter{footnote}{1} \footnotetext{Universit\'e Lyon 1 and CREST.}

\date{}
\maketitle

\begin{center}
\textbf{The Causal-Noncausal Tail Processes: An Introduction} \vspace{1em}

Abstract
\end{center}

This paper considers one-dimensional mixed causal/noncausal autoregressive (MAR) processes with heavy tails, usually introduced to model trajectories with patterns, including asymmetric peaks and troughs, speculative bubbles, flash crashes, or jumps. We especially focus on the extremal behaviour of these processes, when at a given exogenous date the process is above a large threshold, and emphasize the roles of pure causal and noncausal components of the tail process. We provide the dynamic of the tail process and explain how it can be updated during the life of a speculative bubble. In particular, we discuss the prediction of the turning point(s) and introduce pure residual plots as diagnostic tools for the bubble episodes.\vspace{1em}

\textbf{Keywords}: Linear Process, Mixed Causal-Noncausal Process, Conditional Extreme Value, Conditional Tail Process, Conditional Pure Residual Plots, Bubble.

\section{Introduction}
Mixed causal-noncausal autoregressive (MAR) processes are stationary nonlinear processes whose trajectories can feature special patterns as asymmetric peaks and troughs, local trends, speculative bubbles, flash crashes, or jumps. These patterns are typically encountered when analyzing commodity prices, as the oil prices [Lof and Nyberg (2017), Cubbada et al. (2023)], the exchange rates of electronic currencies, as the Theter, or the bitcoin [Hencic and Gouri\'eroux (2019), Cavaliere et al. (2020)], financial indexes as the S\&P 500 or the Nasdaq [Fries (2022)], the evolution of climate risks, as the El Nino and La Nina phenomena [De Truchis, Fries and Thomas (2024)]. {We assume throughout the paper (except in section 5.3) that the MAR process is well-specified.}\footnote{It is out of the scope of the present paper to compare the causal-noncausal modelling with alternative models for bubble based on mildly explosive models introduced in Phillips, Shi and Yu (2015 a, b) [see Phillips and Shi (2020), Skrobotov (2024) for surveys of this literature]. }$^{,}$\footnote{Note that we will avoid a precise definition of a bubble, especially as opposed to fundamental series. The reason is due to recent results on stationary solutions in rational expectation models that include the standard forward solution (usually called fundamental) as well as stationary martingales with bubbles [Gouri\'eroux, Jasiak and Monfort (2020)], such as noncausal MAR. }

Although the estimation methodology of MAR processes has been well documented in the literature,  the prediction of such processes is fairly complicated and often simulation based [Gouri\'eroux and Jasiak (2016)].  This paper considers one-dimensional MAR processes and focuses on these extreme patterns.  In particular,  we introduce the tail process of a MAR process with Paretian (i.e. regular varying) error terms,  and explain how such results can be used to get simple approximations of the predictive distribution of a MAR process during bubble episodes.  In this respect, it complements recent results derived in the special case of MAR processes with $\alpha$-stable distributions\footnote{Such as Cauchy distributions, when $\alpha = 1.$} [Gouri\'eroux and Zako\"ian (2017), Fries and Zako\"ian (2019), Fries (2022),  De Truchis et al. (2025)].  Most of the results derived in this paper rely on different variants of the Single Big Jump (SBJ) heuristic or principle [Lehtomaa (2015),  Kulik and Soulier (2020)].  In its simplest form\footnote{See Feller (1991), Chapter VIII, Proposition on p. 278, or exercise 27 on p. 288.},  it says that,  if $X_1$ and $X_2$ are independent and have Paretian tails,  with survival functions that are asymptotically proportional one to the other\footnote{In this paper, we say that they have equivalent survival functions.  See Section 4.1 for equivalent density functions. }, i.e. : $\lim_{y \rightarrow \infty} \frac{\mathbb{P}[X_1>y]}{\mathbb{P}[X_2>y]}= \xi>0$,  then we have:
$$
\frac{\mathbb{P}[X_1+X_2>y]}{\mathbb{P}[X_1>y]+ \mathbb{P}[X_2>y]} \rightarrow 1,
$$
as $y$ increases to infinity.  This means that an extreme value of the sum $X_1+X_2$ is almost entirely due to one single extreme value of either $X_1$ or $X_2$.  A stronger version (see Lemma 2 in Section 4) is that conditional on $X_1+X_2>y$,  where $y$ is large,  or $X_1+X_2=y$,  with $y$ large,  the conditional distribution of the ratio $\frac{X_1}{X_1+X_2}$ converges to a Bernoulli distribution,  in other words, one of the two terms would be dominating.  These results can be extended to the sum of an arbitrary number of independent variables with equivalent survival functions,  and will allow us to derive quite simple limiting distributions for,  among others,  the predictive distribution of $y_{T+h},  h=1,2,...$,   where $(y_t)$ is a MAR process with a large current value $|y_T|$.

The plan of the paper is as follows.  Section 2 reviews the linear processes with heavy tails and the special case of mixed causal-noncausal autoregressive (MAR$(p,q)$) processes of orders $p$ and $q$. Section 3 analyses the extremal behaviour of these processes when, on a given date, the process is above a large threshold. First, we recall the form of the tail process derived in Kulik and Soulier (2020) for linear processes. This result is then applied to MAR processes. In particular, we emphasize the roles of the tail processes associated with the pure causal and noncausal components and the deterministic recursive equations satisfied by the tail process around the turning point of the underlying bubble. Section 4 considers the extremal behaviour of the MAR process for other types of extreme conditioning set, and discusses the updating of the predictive distribution with respect to the conditioning set. Section 5 introduces the (conditional) pure causal (resp. noncausal) residual plots and their confidence bands and explains how these conditional plots can be used as diagnostic tools to analyse the bubble episodes in the MAR framework.  Section 6 concludes. Proofs and additional examples are provided in appendices and online appendices.

\section{Linear Processes}

This section introduces linear processes with heavy tails and their two-sided moving average representations.

\subsection{Definition}

A (one-dimensional) linear process is a strictly stationary process $\left( y_{t}\right)
_{t\in \mathbb{Z}}$ with a two sided moving-average representation:

\begin{equation}
y_{t}=\sum_{h\in \mathbb{Z}}c_{h}\epsilon _{t-h},  \label{Eq_LinearProcess}
\end{equation}%

\noindent where $\left( \epsilon _{t}\right) _{t\in \mathbb{Z}}$ is a sequence of
independent, identically distributed (i.i.d.) random variables, and $\left(
c_{h}\right) $ is the sequence of moving average coefficients [see
Rosenblatt (2012) for an introduction and properties of linear processes].

Joint conditions on the distribution of $\epsilon _{t}$ and the sequence $%
\left( c_{h}\right) $ are required to ensure the existence of the series in $%
\left( \ref{Eq_LinearProcess}\right) $. They concern the tail index $\alpha $
of error $\epsilon _{t}$ assumed to be regularly varying:

\begin{equation}
\mathbb{P}[ |\epsilon _{t}|>y ] =y^{-\alpha }L\left( y\right), \alpha > 0,
\label{Eq_Slowly_varying}
\end{equation}%
where $L (.)$ is a slowly varying function, the existence of an extremal skewness $\pi$:

\begin{equation}
\label{skewness}
\lim_{y\rightarrow \infty }\frac{\mathbb{P}[ \epsilon _{t}>y]}{\mathbb{P}[|\epsilon _{t}|>y] }=\pi \in (0,1],
\end{equation}%
and the (powered) summability of the $c_{h}$:

\begin{equation}
\sum_{h\in \mathbb{Z}}|c_{h}|^{\delta }<\infty \text{,\quad for }\delta \in
(0,\alpha )\cap (0,1].  \label{Eq_Cond_c_k}
\end{equation}%

Then we have $\mathbb{E}[|\epsilon _{t}| ^{\delta }]<\infty $,
and 
$\left( y_{t}\right) $ is a
well-defined, strictly stationary process, with $\mathbb{E}[| y_{t}|
^{\delta }]<\infty .$

As seen in conditions $\left( \ref{Eq_Slowly_varying}\right) -\left( \ref%
{Eq_Cond_c_k}\right) $, we focus on linear processes when the errors $\left(
\epsilon _{t}\right) $ have Paretian tails\footnote{%
It is, of course, possible to consider processes with thin tails as Gaussian
processes, but such Gaussian processes do not provide the extreme patterns of interest.}. Indeed, with fat tails, we expect trajectories of process $%
\left( y_{t}\right) $ to respond in special ways to the drawing of an error $%
\epsilon _{t}$ in the tail, depending on the sequence $(c_{h})$. More
precisely, such a drawing can create jumps (in the causal case where $%
c_{h}=0 $, $h<0$), speculative bubbles (in the noncausal case where $c_{h}=0$%
, $h>0$), or asymmetric peaks and throughs (in the mixed causal-noncausal case) [see  Gouri\'eroux and Zako\"ian (2019) for a discussion]. We are especially interested in these extreme patterns.

It is known [Rosenblatt (2012)], that the representation $\left( \ref%
{Eq_LinearProcess}\right) $ of a linear process is not unique. The
distribution of $\epsilon _{t}$ and the sequence $(c_{h})$ are defined up to a signed scalar, and to the choice of a maturity origin, if at least one $c_{h}$ is
non zero. Then we can identify  the linear representation by imposing the maturity origin $0\in \mathbb{%
Z}$, such that $c_{0}=1$, for instance. Later on we assume simply $c_{0}>0$.\footnote{%
There exist other identification issues in the Gaussian case that we do not
discuss in this paper.}

\subsection{Mixed Causal-Noncausal Autoregressive Model}

It is usual to consider linear processes satisfying a mixed autoregressive
(MAR) specification [Lanne  and Saikkonen (2011),  Fries and Zako\"ian (2019)] of the type :
\begin{equation}
	\label{general}
\Phi \left( L\right) \Psi \left( L^{-1}\right) y_{t}=\epsilon _{t},
\end{equation}%
where:
\begin{eqnarray}
\Phi \left( L\right) &=&1-\phi _{1}L-...-\phi _{p}L^{p}, \\
\Psi \left( L\right) &=&1-\psi _{1}L-...-\psi _{q}L^{q},
\end{eqnarray}%
$L$ denoting the lag (backward) operator and $L^{-1}$ the lead (forward) operator. The roots of the operators $\Phi $ and $\Psi $
are strictly outside the unit circle. Such a representation is denoted MAR($%
p,q$), where $p$ and $q$ are the causal and noncausal orders, respectively.\footnote{Fries and Zako\"ian (2019) denote such a process MAR($q, p$),  instead.  We follow the initial notation of Lanne  and Saikkonen (2011).} We assume $\varphi_p \neq 0$, $\psi_q \neq 0$ for the orders to be uniquely defined. This specification has implicitly introduced identification restrictions on the autoregressive coefficients by assuming $ 
\Phi \left( 0\right) =1$ and $\Psi \left( 0\right)=1.$  
Such a MAR ($p,q$) model can be equivalently written as:%

\begin{equation}
\prod_{i=1}^{p}\left( 1-\lambda _{i}L\right) \prod_{j=1}^{q}\left( 1-\mu
_{j}L^{-1}\right) y_{t}=\epsilon _{t},
\end{equation}%

\noindent where $\lambda _{i},i=1,...,p,$ and $\mu _{j},j=1,...,q,$ are the inverses of the roots of polynomials $\Phi $ and $\Psi $, respectively. These roots can be real or complex, single or multiple. By definition, we have:

\begin{equation*}
|\lambda _{i}|<1,\forall i,|\mu _{j}|<1,\forall j.
\end{equation*}%

The strictly stationary solution of $\eqref{general} $ is unique
and admits a two sided moving average representation in $\left( \epsilon
_{t}\right) $ obtained by inverting the operators $\Phi \left( L\right) $ and $%
\Psi \left( L^{-1}\right) $. More precisely, we can write:%

\begin{equation}
y_{t}=\frac{1}{\Phi \left( L\right) \Psi \left( L^{-1}\right) }\epsilon _{t},
\end{equation}%
where:%

\begin{eqnarray}
\frac{1}{\Phi \left( L\right) } &=&\frac{1}{\prod_{i=1}^{p}\left( 1-\lambda
_{i}L\right) } \equiv \sum_{h=0}^{\infty }a_{h}L^{h}, \\
\frac{1}{\Psi \left( L^{-1}\right) } &=&\frac{1}{\prod_{j=1}^{q}\left(
1-\mu _{j}L^{-1}\right) }\equiv \sum_{h=0}^{\infty }b_{h}L^{-h}.
\end{eqnarray}%

Then we deduce:%

\begin{equation*}
y_{t}=\left( \sum_{h=0}^{\infty }a_{h}L^{h}\right) \left( \sum_{h=0}^{\infty
}b_{h}L^{-h}\right) \epsilon _{t}= \sum_{h\in \mathbb{Z}}c_{h}L^{h}\epsilon
_{t},
\end{equation*}%

\noindent where $c_{h}=\sum_{k=h^{+}}^{\infty }a_{h} b_{k-h},$ with $h^{+}=\max \left(
h,0\right) $.

Let us now discuss the closed form expressions of the sequences $a_{h}, b_{h},c_{h}$ as
functions of the $\lambda _{i},\mu _{j}$. For expository purposes, let us focus on polynomials of degree non larger than $2$, the general case being analyzed in Appendix A.2.

We have:%
\begin{eqnarray*}
\frac{1}{\left( 1-\lambda _{1}L\right) \left( 1-\lambda _{2}L\right) } &=&%
\frac{1}{\lambda _{2}-\lambda _{1}}\left( \frac{\lambda _{2}}{1-\lambda _{2}L%
}-\frac{\lambda _{1}}{1-\lambda _{2}L}\right) \\
&=&\frac{1}{\lambda _{2}-\lambda _{1}}\left( \sum_{h=0}^{\infty }\left(
\lambda _{2}^{h+1}-\lambda _{1}^{h+1}\right) L^h\right) .
\end{eqnarray*}%

Then we deduce:%

\begin{equation}
a_{h}=\frac{\lambda _{2}^{h+1}-\lambda _{1}^{h+1}}{\lambda _{2}-\lambda _{1}}%
,
\end{equation}%
and similarly:

\begin{equation}
\label{expressionbh}
b_{h}=\frac{\mu _{2}^{h+1}-\mu _{1}^{h+1}}{\mu _{2}-\mu _{1}}.
\end{equation}

These expressions are always valid, but can be particularized for conjugate
complex roots or for double real root. In fact, we have the following three cases:

\begin{enumerate}
\item When $\lambda _{1}$ and $\lambda _{2}$ are real and distinct: we have:
\begin{equation*}
a_{h}=\frac{\lambda _{2}^{h+1}-\lambda _{1}^{h+1}}{\lambda _{2}-\lambda _{1}}%
;
\end{equation*}

\item When $\lambda _{1}$ and $\lambda _{2}$ are complex conjugate pairs: $\lambda
_{1}=\rho \exp \left( i\omega \right)$, with $\rho>0$, we have:

\begin{equation*}
a_{h}=\rho ^{h}\frac{\sin \left( \omega \left( h+1\right) \right) }{\sin
\left( \omega \right) };
\end{equation*}

\item When $\lambda _{1}$ is a double real root:%
\begin{equation*}
a_{h}=\lambda _{1}^{h}\left( 1+h\right) .
\end{equation*}
\end{enumerate}

Let us provide below the closed form of the associated moving average coefficients $c_{h}$ for different MAR($p,q$) processes (with real roots):

\begin{description}
\item[-] MAR($1,0$) = AR($1$)%
\begin{equation*}
c_{h}=\lambda ^{h},h\geq 0,c_{h}=0,h<0.
\end{equation*}

\item[-] MAR($2,0$) = AR($2$) (distinct root)%
\begin{equation*}
c_{h}=\frac{\lambda _{2}^{h+1}-\lambda _{1}^{h+1}}{\lambda _{2}-\lambda _{1}}%
,h\geq 0,c_{h}=0,h<0.
\end{equation*}

\item[-] MAR($2,0$) = AR($2$) (double root)%
\begin{equation*}
c_{h}=\lambda ^{h}\left( 1+h\right) ,h\geq 0,c_{h}=0,h<0.
\end{equation*}

\item[-] MAR($0,1$)%
\begin{equation*}
c_{h}=0,h>0,c_{h}=\mu ^{-h},h\leq 0.
\end{equation*}

\item[-] MAR($0,2$) (distinct roots)%
\begin{equation*}
c_{h}=0,h>0,c_{h}=\frac{\mu _{2}^{-h+1}-\mu _{1}^{-h+1}}{\mu _{2}-\mu _{1}}%
,h\leq 0.
\end{equation*}

\item[-] MAR($1,1$)%
\begin{equation*}
c_{h}=\frac{1}{1-\lambda \mu }\lambda ^{h},h\geq 0,c_{h}=\frac{1}{1-\lambda
\mu }\mu ^{-h},h\leq 0.
\end{equation*}
\end{description}

\subsection{Positivity}

The two-sided moving average representation can be applied with $\epsilon _{t}$ following a continuous distribution on $( -\infty, \infty) $, and with moving average coefficients $c_{h}$ of any sign, except $c_{0}$, assumed strictly positive for identification. However, these representations are often applied to series $\left(
y_{t}\right) $ that take positive values as commodity prices [Gouri\'eroux and
Zako\"ian (2019), Fries (2021)],  traded volumes [Bilayi-Biakana et al.
(2019)], or the air pressure differential for the analysis of El Nino-La
Nina phenomena [De Truchis et al. (2024)].

Then the representation can be applied as either:

\begin{equation}
y_{t}=\sum_{h\in \mathbb{Z}}c_{h}\epsilon _{t-h},  \label{Eq_Positive_eps}
\end{equation}%
or:

\begin{equation}
\label{log}
\log y_{t}=\sum_{h\in \mathbb{Z}}\tilde{c}_{h}\tilde{\epsilon}%
_{t-h}\Leftrightarrow y_{t}=\exp \left( \sum_{h\in \mathbb{Z}}\tilde{c}_{h}%
\tilde{\epsilon}_{t-h}\right) .
\end{equation}%

The models \eqref{Eq_Positive_eps}-\eqref{log} are not compatible except in the white noise case $c_{h}=0$, $%
\forall h\neq 0$.

It is easily checked that the first representation $\left( \ref%
{Eq_Positive_eps}\right) $ of the positive series $\left( y_{t}\right) $
implies that the distribution of $\epsilon _{t}$ is on $\left( 0,\infty
\right)$, and all the coefficients $c_{h}$ are nonnegative (by the
identifying restriction $c_{0}>0$). Then the two representations, i.e. the linear one and the exponential one, will lead to trajectories with different
patterns. Indeed, when $\epsilon _{t}$ is constrained to be positive, we can
only have right fat tail effects, whereas left and right fat tail effects
can exist in the exponential model. This difference of patterns depend on
the variables and phenomena of interest. Typically, some left tail effects
can be interesting to capture, as flash crashes in market prices, or
liquidity gaps in traded volumes.

The nonnegativity of the coefficient $c_{h}$ implies restrictions on the
coefficients of the MAR representation, i.e. on $\lambda _{i},\mu _{j}$. In
our examples with $p,q$ (smaller or) equal to $2$, we see that this
nonnegativity condition is realized if and only if $\lambda _{1},\mu
_{1,}\lambda _{2},\mu _{2}$ are all real nonnegative, including the
possibility of double roots. It also induces restrictions on the coefficients $\varphi,  \psi$ as alternating signs: $\varphi_1 > 0, \varphi_2 < 0, \psi_1 >0, \psi_2 < 0$.

\section{Extremal Behaviour And Tail Process}

If we consider an extreme impulse $\delta $ on $\epsilon _{t}$, the
responses on $y_{t+h}$ will be $c_{h}\delta $, $h\in \mathbb{Z}$. Therefore, 
the sequence of $c_{h}$'s could be considered as the Impulse Response
Function (IRF) corresponding to a unitary shock on $\epsilon_t$. However, this interpretation is a bit misleading if the large
value of $\epsilon _{t}$ is not controlled, but results from a drawing in
the tail. More precise statements are obtained in a conditional extreme
value framework [Kulik and Soulier (2020), Section 15.3].\footnote{See also Giancaterini,  Hecq,  Jasiak and Neyazi (2025) for a recent application of this result to the MAR(1,1) framework.}

\subsection{The Conditional Extreme Value (CEV) Framework}

We have the following proposition obtained by applying Proposition 5.2.5 of
Kulik and Soulier (2020) to infinite moving-average processes.

\begin{Proposition} The conditional distribution of the process $(y_{t+h}/|y_{t}|)_h$, $h$ varying, conditional on $|y_{t}|>y$, converges to the distribution of a process $\left(
X_{h}\right) $ when $y$ tends to infinity:%

\begin{equation*}
\mathcal{L}\left(  \Big( \frac{y_{t+h}}{|y_{t}|} \Big)_h \mid |y_{t}|>y\right) \overset{%
\mathrm{d}}{\longrightarrow }\mathcal{L}\Big( \left( X_{h}\right)_h \Big) ,
\end{equation*}%

\noindent where $X_{h}=X_{0}c_{N+h}/c_{N}$, $X_{0}$ and $N$ being two independent random variables, $X_{0}$ is a variable taking values $+1,-1$ with
probabilities $\pi ,1-\pi $, where the extremal skewness $\pi$ is defined in eq. \eqref{skewness}, %
and $N$ is a discrete random variable whose distribution is defined by:
\begin{equation}
p_{j}:=\mathbb{P}[ N=j] =\frac{|c_{j}|^{\alpha }}{\sum_{n\in \mathbb{Z}}|c_{n}|^{\alpha }}%
,\quad j\in \mathbb{Z}.
\end{equation}

The process $\left( X_{h}\right) $ is called the (spectral) tail process [see Basrak and Segers (2009),  Section 1].\end{Proposition}

Here,  the convergence of the process $(y_{t+h}/|y_t|)$ means that any finite dimensional joint distribution of this process converges to the corresponding finite dimensional joint distribution of process $(X_h)$.  Note that the tail process is indexed by the horizon $h$ from date $t$, not by the calendar time $t$. Thus, the distribution of the tail process does not depend on $t$.  Note also that it depends on the distribution of $(\epsilon_t)$ through the tail parameters $\alpha$ and $\pi$.  In other words,  even if the initial process is semiparametric in nature when the distribution of $(\epsilon_t)$  is unspecified, it becomes parametric in the CEV framework. 

This proposition says that  when correctly normalized (by $|y_{t}|$),  the distribution of the process $(y_{t+h})$ is discrete,  making it ``easier" to predict future values of the process.  Moreover,  variables $X_0$ and $N$ can be interpreted using the single jump heuristic.  Indeed,  first,  given $|y_t|$ is large,  we can either have $y_t>y$,  or $y_t< -y$, with approximate probabilities $\pi$ and $1-\pi$, respectively,  where $\pi= \lim_{y \to +\infty} \mathbb{P}[y_t>y \mid |y_{t}|>y]$.  Thus $X_0$ is the indicator variable deciding which one of these two cases arises.  Let us, without loss of generality, focus on the case where $y_t>y.$  

Because $y_t= \sum_{j} c_j \epsilon_{t-j}$ is a linear combination of independent variables, all with equivalent tails,  an infinite term extension of the single jump heuristic mentioned in the Introduction says that exactly one among the $ c_j \epsilon_{t-j}$'s, $j$ varying,  will be large,  with the probability that it is  $c_j \epsilon_{t-j}$ being $p_j$.  Thus $t-N$ is the stochastic index of the large jump,  with the convention that $N>0$ indicates the large jump $\epsilon_{t-N}$ corresponds to a past date. That is, $N$ measures the distance between the current time $t$ and the location of the big jump,  called the ``epicenter". Then, one has:
$$
y_t \approx c_N \epsilon_{t-N}, \qquad y_{t+h} \approx c_{N+h} \epsilon_{t-N},  
$$  
and hence $y_{t+h}/y_t \approx c_{N+h}/c_N$. 

The result is greatly simplified when it is applied to a positive series $%
\left( y_{t}\right) $, i.e., when $\left( \epsilon _{t}\right) $ as well as
the $c_{h}$'s are nonnegative. Indeed, we get:

\begin{equation*}
X_{h}=c_{N+h}/c_{N},h\in \mathbb{Z},
\end{equation*}%
with $N$ defined in the same way as in Proposition 1.  

Contrary to the naive derivation of the IRF given by $c_{h}$ (which
corresponds to $N=0$ under the identification restriction $c_{0}=1$), we note
that the tail process involves a change of time origin by the
stochastic drift $N$. Under conditional extreme value, the variables $X_{h}$,
$h\in \mathbb{Z}$, depend non linearly on the unique stochastic factor $N$,
which can create deterministic links between the $X_{h}$'s, and then a dynamics different from the unconditional dynamics of $\left( y_{t+h}/|y_{t}|\right) $.

Alternatively, Proposition 1 can be written for the relative changes:

\begin{equation*}
r_{t+h}= y_{t+h}/y_{t+h-1}.
\end{equation*}%

We get the following result:\vspace{1em}

\begin{Proposition} For a positive process $(y_t)$, the conditional distribution given $y_{t}>y$ of the process $(r_{t+h})$, $h$ varying, converges to the distribution of process $$\left(
Z_{h}\right) =\left( X_{h}/X_{h-1}\right)=\big(c_{N+h}/c_{N+h-1}\big),$$ %
\noindent when $y$ tends to infinity, where $N$ and $\left( X_{h}\right) $ are defined as in Proposition 1.\vspace{1em}
\end{Proposition}

Note that Proposition 1 (resp. Proposition 2) implicitly assumes that $c_{N}$
(resp. all $c_{N+h-1}$) is not equal to zero with a strictly positive
probability.  
\vspace{1em}

\begin{Remark} When the error $\epsilon_t$ has an $\alpha-$stable distribution and for special MAR processes, the asymptotic behaviours of $(y_t)$ conditional on $y_t > y$ have been derived directly by using the spectral representation of the multivariate $\alpha$-distribution [see Rootzen (1978), Samorodnitsky and Taqqu (1994) for this representation, and Fries (2022), Section 4, De Truchis et al. (2025), for its use in the MAR framework]. The approach based on Proposition 1 is much more general and allows to working directly with the tail process, as seen in the next sections.\end{Remark}

\begin{Remark} Whereas the trajectories of $(r_{t+h})_h$ can take very different patterns, this is not the case for their tail analogue, which  weights only a countable set of patterns.
\end{Remark}

Let us now consider a MAR$(p,q)$ model with a parametric error distribution and denote by $\theta$ the vector of parameters that include the AR coefficients and the tail parameter of the error distribution. Then the distribution of the stochastic drift $N$ and the support of the tail process $(Z_h)$ depend on $\theta$. In practice, $\theta$ is unknown, but can be estimated by approximate maximum likelihood from a batch of data $y_t, t=1,\ldots, T$. If the model is well-specified, this estimator is consistent, converges at speed $\sqrt{T}$, and is asymptotically normal [Breidt, David, Lii and Rosenblatt (1991), Lanne and Saikkonen (2011), Davis and Song (2020)]. Therefore, all these summaries of the distribution of the tail process can be estimated by plugging in the estimator $\hat{\theta}_T$ instead of the unknown value $\theta$. These estimated summaries will converge to their asymptotic counterparts and their asymptotic distribution will be derived by the delta-method, if these summaries are differentiable functions of $\theta$.\footnote{Note that this delta method cannot be applied to the estimated support of the tail process itself.} This estimation can also be performed in a semi-parametric framework.  If $\Phi$ and $\Psi$ are estimated by the generalized covariance estimator [Gouri\'eroux and Jasiak (2023)],  or by minimizing an integrated residual dependence measure [Valesco (2022)], then we can derive approximated errors $\widehat{\epsilon_t}$ and deduce from this sample of residuals estimators of $\alpha$ and of the density of $\epsilon_t$.

\subsection{Applications}

Let us now apply Proposition 1 and/or Proposition 2 to some MAR examples of
Section 2.2. For this illustration, we consider positive processes, that is $\lambda_i>0, i=1,...,p,   \mu_j >0,  j=1,...,q$, and a conditioning performed at a given date $t$ corresponding to residual maturity $0$. Additional examples are provided in Appendix A.2.\vspace{1em}

\subsubsection{Pure causal process of order $1$: MAR($1,0$) = AR($1$)}

In this example, the stochastic drift $N$ is nonnegative, which means that the SBJ $\epsilon_{t-N}$ is indexed by a past or current date. Moreover,  we have:%

\begin{equation*}
\mathbb{P}[ N=j] =p_{j}=\left( 1-\lambda ^{\alpha }\right) ^{-1}\lambda
^{\alpha j},j\geq 0.
\end{equation*}%

Thus $N$ follows a Pascal (geometric) distribution.  Then we have:

\begin{equation*}
X_{h}=X_0\lambda ^{N+h}/\lambda ^{N}=\lambda ^{h},h\geq 0.
\end{equation*}%

In other words,  for $h\geq 0$, the tail process is no longer stochastic.  This can be explained as follows.  Since $Y_t= \lambda Y_{t-1}+\epsilon_t$,  conditional on $|y_t|>y$,  we have approximately $y_{t+1}/y_t \approx y_{t+2}/y_{t+1} \approx \cdots \approx \lambda$.  In other words, upon normalization,  the future trajectory becomes deterministic.   

However,  the effect of the stochastic drift appears for negative $h$, since:%

\begin{equation}
X_{h}=X_0\frac{\lambda ^{N+h}}{\lambda ^{N}}\mathbb{I}_{N+h\geq 0}=\lambda ^{h}%
\mathbb{I}_{N\geq -h},h<0,
\end{equation}
that corresponds to a backward binomial tree. This simply means that when we look backward in time,  since $y_t= \lambda y_{t-1}+\epsilon_t$,   conditional on $|y_t|>y$,   the previous observation $y_{t-1}$ can be either very large (and approximately equal to $y_t/\lambda$,  or close to zero,  according to the SBJ described in the Introduction.  Similarly,  conditional on $y_{t-1}$ being large,  $y_{t-2}$ can be even larger,  or close to zero.  Hence the backward binomial tree.

\subsubsection{Pure noncausal process of order $1$: MAR($0,1$)}

The situation is symmetric for the pure noncausal process. The stochastic drift is nonpositive:
\begin{equation*}
X_{h}=X_0\mu ^{-h},h\leq 0,
\end{equation*}%
meaning that the SBJ is indexed by the current or a future date.  Moreover, we have:%
\begin{equation}
X_{h}=X_0\mu ^{-h}\mathbb{I}_{N\leq -h},h>0.
\end{equation}

Thus we have a deterministic evolution (up to the sign $X_0$) for negative $h$ and a forward binomial tree for $h>0$.  In particular,  conditional on $|y_t|>y$,  $y_{t+1}/y_t$ is either close to $1/\mu$ (corresponding to further accumulation of the bubble),  or close to zero (corresponding to the collapse of the bubble).  Moreover,  in case $y_{t+1}/y_t$ is close to $1/\mu$,  the same bi-modal pattern can be said of $y_{t+2}/y_{t+1}$,  and possibly also $y_{t+3}/y_{t+2}, ...$,  hence the forward binomial tree.  This property has previously been derived in Fries (2022).  

\subsubsection{Pure noncausal MAR($0,2$) with double root}

We get:

\begin{equation}
X_{h}=X_0\mu ^{-h}\left( 1-h\right) \mathbb{I}_{N\leq -h},\forall h,
\end{equation}%

\noindent with a nonpositive $N$ factor. As in Section 3.2.2, we get a deterministic evolution (up to  the sign $X_0$) before $h=0$ (date $t$)
and an evolution with binomial tree after $h=0$ (date $t$). The difference
with Section 3.2.2 is in the explosion rate in the explosive branch. In both
Sections 3.2.2 and 3.2.3., the opposite of the stochastic drift $-N$ gives the
stochastic maturity until the bubble crash, that is the transition to value zero.

In the pure causal, or noncausal cases, the tail process $\left(
X_{h}\right) $ can take the value zero for some $h$ with strictly positive
probability. In this case the asymptotic changes $Z_{h}$ are not always
defined and Proposition 2 cannot be applied. This is no longer the case for
mixed models.\vspace{1em}

\begin{Remark} The stochastic binomial trees play a special role in Finance, where they are used as approximations of continuous time models by state and time discretizations, such as the Black-Scholes model [Cox, Ross and Rubinstein (1979)]. Such interpretations are valid for pure processes of order 1. In our framework with fat tails, the standard Black-Scholes equations will have to be replaced by a stochastic differential equation $dy_t = \mu y_t dt + \sigma y_t d \mathcal{L}_t$, where $\mathcal{L}_t$ denotes a Levy process, and then an extremal behaviour of $y_{t+h}/y_t$ given $y_t > y$ with $y$ large, will lead to this type of tree.
\end{Remark}

\subsubsection{Mixed causal-noncausal process MAR($1,1$)}

In this case the drift $N$ can take positive as well as negative values and
it is easier to work with the tail variables $Z_{h}$. We have:%
\begin{equation}
\label{mar110}
c_{h}/c_{h-1} =
\begin{cases}
& \lambda ,\quad \text{if }h\geq 1, \\
&\mu ^{-1},\quad \text{if }h\leq 0.
\end{cases}
\end{equation}
We deduce that:
\begin{equation}
\label{mar11}
Z_{h}  = c_{N+h}/c_{N+h-1}  =\mu ^{-1}\mathbb{I}_{N\leq -h}+\lambda \mathbb{I}_{N\geq 1-h}.
\end{equation}

In the mixed case, we get a binomial tree in terms of changes with a branch
exploding at rate $\mu ^{-1}$\ with probability $\mathbb{P}[N\leq -h]$,
and another branch decreasing at rate $\lambda $ with probability $\mathbb{P}[N\geq 1-h] =1-\mathbb{P}[N\leq -h]$.

It is also easily checked that the distribution $p_j = \mathbb{P}[N=j]$ of the stochastic drift $N$  is given by:
$$
p_j =\begin{cases}
\Big[\frac{1}{1-\lambda^{\alpha}} + \frac{1}{1-\mu^{\alpha}} - 1\Big]^{-1} \mu^{-\alpha j}, &\; \mbox{if}\; j \leq 0, \\
\Big [\frac{1}{1-\lambda^{\alpha}} + \frac{1}{1-\mu^{\alpha}} - 1\Big]^{-1} \lambda^{\alpha j}, &\; \mbox{if}\; j \geq 0,
\end{cases},
$$
where the two formulas coincide for $j=0$.  Thus the distribution of $N$ is a mixture of two types of ``geometric" distributions, that are a standard one for $j \geq 0$ and another written in reversed time for $j < 0$. This mixture distribution is symmetric if and only if the causal and noncausal roots are equal. The mode of the distribution is for $N=0$, but its mean can be of any sign depending if $\lambda$ is larger, or smaller than $\mu$ [see also Giancaterini et al. (2025)].

\subsubsection{Pure noncausal MAR($0,2$) with distinct roots}
Let us assume that $\mu_1$ and $\mu_2$ are distinct and positive, with $\mu_1>\mu_2$.  Then we have $c_{h+1}/c_{h}=\frac{\mu_1^{h+2}-\mu_2^{h+2}}{\mu_1^{h+1}-\mu_2^{h+1}}$ for any $h\geq 0$,  and it is easily checked that this sequence is decreasing,  with limit $\mu_1$.  Thus,  the typical future trajectory of an AR(2) process during a bubble is the following: when the current time $t$ is still far from the ``epicenter" of the bubble,  the process increases at approximately the rate $\mu_1^{-1}$; then the rate of accumulation decreases to its minimum value $b_1^{-1}=(\mu_1+\mu_2)^{-1}$,  and then the bubble collapses.  Note that this bubble pattern is almost deterministic,  on the contrary to noncausal AR(1) process.  This property has also been obtained by De Truchis et al.  (2025,  Remark 4.2) in the special case of $\alpha-$stable processes.  

\subsection{Causal and Noncausal Components}

To understand the behaviour of the tail process for a general mixed autoregressive process, it is useful to consider the causal and noncausal components of the MAR process [Lanne and Saikkonen (2011)]. Let us consider a MAR process $(y_t)$ such that:
\begin{equation*}
  \Phi (L) \Psi (L^{-1}) y_t = \epsilon_t \Longleftrightarrow y_t = \Sum^{+ \infty}_{h=-\infty} c_h \epsilon_{t-h}.
\end{equation*}
Its pure noncausal component $(u_t)$ is defined by:
\begin{equation}
	\label{ut}
u_t = \Phi (L) y_t = \frac{1}{\Psi (L^{-1})} \epsilon_t \equiv \Sum^0_{h=-\infty} b_h \epsilon_{t-h}.
\end{equation}
Its pure causal component $(v_t)$ is defined by:
\begin{equation}
	\label{vt}
v_t = \Psi (L^{-1}) y_t = \frac{1}{\Phi (L)} \epsilon_t = \Sum^\infty_{h=0} a_h \epsilon_{t-h}.
\end{equation}
The lemma below explains how to derive the moving-average coefficients $(a_h), (b_h)$ of the pure components from the moving average coefficients $(c_h)$ of the process. We first denote by $\tilde{L}$ the lag (backward) operator on maturities $h$.  That is,  for any sequence $(c_h)$,  $\tilde{L}(c_h)=c_{h-1}$,  for any $h \in \mathbb{Z}$.  

\begin{lm}  We have:\vspace{1em}

i) $b_h = \Phi (\tilde{L}) c_h, \forall h.$

ii) $a_h = \Psi (\tilde{L}^{-1}) c_h, \forall h.$

iii) In particular we get:

$$
\Phi (\tilde{L}) c_h  =  0, \forall h \geq 1, \qquad \Psi (\tilde{L}^{-1}) c_h  =  0, \forall h \leq -1.
$$
\end{lm}
\begin{proof}
See Appendix A.1.
\end{proof}

Let us now assume $y_t > y$ and standardize by $y_t$ the pure causal and noncausal components. We have, for large $y$:

\begin{eqnarray}
  u_{t+h}/y_t & =& [\Phi (L) y_{t+h}]/y_t \stackrel{d}{\rightarrow} \Phi (\tilde{L}) X_h = X_0\Phi (\tilde{L}) \frac{c_{h+N}}{c_N}, \\
  \mbox{and}\; \; v_{t+h}/y_t & = & [\Psi (L^{-1}) y_{t+h}]/y_t \stackrel{d}{\rightarrow} \psi (\tilde{L}^{-1}) X_h = X_0\Psi (\tilde{L}^{-1}) \frac{c_{h+N}}{c_N}.
\end{eqnarray}

By applying Lemma 1 iii), we get the following: \vspace{1em}

\begin{Proposition} The tail process $(X_h)$ satisfies the deterministic recursion:

$$
\begin{array}{lcl}
\Phi (\tilde{L}) X_h & = & 0, \; \mbox{if}\; h+N \geq 1, \\ \\
\Psi (\tilde{L}^{-1}) X_h& =& 0, \; \mbox{if}\; h+N \leq -1.
\end{array}
$$
\end{Proposition}

Note that,  for given $N$,  these deterministic recursive equations do not depend on the distribution of the error $\epsilon_t$.

By eq. \eqref{ut} [resp.  eq.  \eqref{vt})], ${U}_h =\Phi (\tilde{L}) X_h$ (resp. ${V}_h =$ $\Psi (\tilde{L}^{-1}) X_h)$ can be interpreted as the pure noncausal (resp. pure causal) components of the tail process.  Note that the pure tail processes depend on the exogenous date $t$,  but also on the MAR($p,q$) process assumed to be well-specified (see Section 5.3 for a discussion of mis-specified pure tail processes).  We deduce the following corollary:

\begin{Corollary} The pure noncausal tail process ${U}_h$ is zero, if $h \geq 1-N$.

\hspace{1,5cm} The pure causal tail process  ${V}_h$ is zero, if $h \leq -1-N$.\end{Corollary}

Let us consider a series of positive observations $y_t, t=1, \ldots, T$,  suppose that $y_T > y$, and consider the tail process associated with this date $T$. Then we have three regimes:

i) a pure causal regime, if $h \geq 1-N$.  Under this regime,  the right side of the tail process $(X_h)_{h \geq 1-N}$ is (deterministically) Markov of order $p$; 

ii) a pure noncausal regime, if $h \leq -1 - N$. Under this regime,  the left side of the tail process $(X_h)_{h \leq -1-N}$ is deterministically Markov of order $q$; 

iii) a mixture of the causal and noncausal regimes if $h =-N$ (that corresponds to the stochastic time index $T-N$ in the underlying calendar time).\vspace{1em}

The above switching regimes interpretation shows that it can be informative in practice to plot not only the trajectory of the observations $y_t$, but also:

 i) the trajectories of the pure causal and noncausal components $\hat{u}_t, \hat{v}_t$ (estimated by replacing the parameters by consistent estimates), and

 ii) at given date $T$ with $y_T > y$ large, the trajectories in maturity $h$ of $\hat{u}_{T+h}/y_T, \hat{v}_{T+h}/y_T, h$ varying, as a descriptive tool to detect the turning point (see Section 5).
 \subsection{One sided Tail Processes}

In some other cases,  it could also be interesting to investigate the tail process based on other decompositions of $y_t$ in terms of $u_t$,  $v_t$.

For a general MAR($p, q$),  the causal-noncausal decomposition of the MAR(1,1) process given in eq. \eqref{yt} can be extended to [see Gouri\'eroux, Jaskak (2016),  section 2.3]:
\begin{equation}
\label{pqvu}
y_t= L^q b_1(L) v_t+ b_2(L) u_t,
\end{equation}
where polynomials $b_1$,  $b_2$ are non zero,  and of degree non larger than $p-1$ and $q-1$,  respectively:
$$
b_1(L)=\sum_{i=0}^{q-1} b_{1,i}L^i, \qquad b_2(L)=\sum_{j=0}^{p-1} b_{2,j}L^j.
$$
These two polynomials result from the partial fraction decomposition:
$$
\displaystyle \frac{1}{\Phi(L) [L^p \Psi(L^{-1})]}= \frac{b_1(L)}{\Phi(L)}+ \frac{b_2(L)}{L^p \Psi(L^{-1})}.
$$
 This decomposition separates the two-sided moving average representation for $y_t$ into two terms,  with the first term $y_{1,t}:=L^q b_1(L) v_t=\sum_{h=1}^{\infty} c_h \epsilon_{t-h}$ containing only past $\epsilon$'s,  and the second term $y_{2,t}:=b_2(L) u_t=\sum_{h=-\infty}^{0} c_h \epsilon_{t-h}$ containing current and future $\epsilon$'s.  Then one can also define the tail processes of the pure one-sided causal process $(y_{1,t})$ and the pure one-sided noncausal process $(y_{2,t})$, respectively.  By Proposition 1,  we get:
\begin{align*}
&\mathcal{L}\Big( \frac{(y_{1,t+h})}{|y_{1,t}|} \mid |y_{1,t}|> y \Big) \rightarrow (X_{1,h}),
&\mathcal{L}\Big( \frac{(y_{2,t+h})}{|y_{2,t}|} \mid |y_{2,t}|>y \Big) \rightarrow (X_{2,h}),
\end{align*}
as $y$ increases to infinity.  We call them one-sided tail processes,  
since $i)$ they are the weak limits of one-sided infinite moving averages $(y_{1,t})$ and $(y_{2,t})$,  respectively; $ii)$, their distributions are such that: $X_{1,h}/X_0$ is deterministic for $h \leq 0$,  whereas $X_{2,h}/X_0$ is deterministic for $h > 0$ (see Sections 3.2.1-3.2.3 for similar properties for the tail process of a one-sided causal or noncausal AR process).  Thus,  both are one-sidedly deterministic.  The interpretation of this result is that $(X_{1,h})$ is ``easy" to predict in the reverse time direction (for $h \leq 0$),  whereas $(X_{2,h})$ is easy to predict in the calendar time direction.  

 
These one-sided tail process differ from the pure tail causal and noncausal components of Section 3.3. In particular, they have the property that the two one sided tail processes are independent.

Like $(X_h)$,  the one-sided tail processes $(X_{1,h})$ and $(X_{2,h})$ have their specific stochastic drifts $N^*_1>0$ and $N^*_2 \leq 0$,  which provide the index of the SBJ among past (resp. current and future) $\epsilon_t$'s,  respectively. Then  the distribution of $N^*_1$,  $N^*_2$ and $N$ are related through: 
$$
N=_{(d)} \eta N^*_1+ (1-\eta) N^*_2,
$$  
where $=_{(d)} $ denotes the equality in distribution,  $\eta$ is a Bernoulli variable with probability parameter $\frac{\sum_{h=1}^{\infty} |c_h|^{\alpha}}{\sum_{h=-\infty}^{\infty} |c_h|^{\alpha}}$ and is independent of $N^*_1$ and $N^*_2$. Indeed the SBJ among all $\epsilon$'s can correspond either to a past date,  or a future (including current) date,  as determined by the realization of the Bernoulli variable $\xi$.  Here,  we are using the fact that given $y_t=y_{1,t}+y_{2,t}$ is large,  exactly one among $y_{1,t}$ and $y_{2,t}$ is large,  and $\xi=1$ if and only if $y_{1,t}$ is large,  and $\xi=0$,  otherwise.  

Similarly,  we can write the ratio $y_{t+h}/y_t$ as:
$$
\frac{y_{t+h}}{y_t}= \frac{y_{1,t}}{y_{1,t}+y_{2,t}} \frac{y_{1,t+h}}{y_{1,t}}+  \frac{y_{2,t}}{y_{1,t}+y_{2,t}} \frac{y_{2,t+h}}{y_{2,t}}.
$$
and deduce:
\begin{equation}
	\label{Xh}
X_h=_{(d)} \eta X^*_{1,h}+ (1-\eta) X^*_{2,h},  \forall h \in \mathbb{Z},
\end{equation}
where $(X^*_{1,h})$ and $(X^*_{2,h})$ are independent copies of $(X_{1,h})$ and $(X_{2,h})$ and are mutually independent. 


This decomposition has an analog in terms of the causal and noncausal components of $(X_h)$.  Indeed,  from eq. \eqref{pqvu},  we deduce that:
\begin{equation}
\label{pqvux}
X_h= L^q b_1(\tilde{L}) V_h+b_2(\tilde{L}) U_h.  
\end{equation}


\subsection{The Turning Point($s$)}

Let us assume a nonnegative sequence $(c_{h})$, $h\in \mathbb{Z}$, with a
unique maximum at $h_{0}=\arg \max_{h}c_{h}$. Then the tail process $\left(
X_{h}\right) $ takes its maximum value when $N+h=h_{0}$, that is at $%
h_{N}=h_{0}-N$. This stochastic maturity provides the date $T+h_{N}$ of the
turning point of the bubble. Before this date, the process is locally in an
increasing phase and decreases after this date.

In the examples of MAR($1,0$), MAR($0,1$), MAR($1,1$), with positive $\mu $
and $\lambda $, we have $h_{0}=0$ and $h_{N}=-N$. However in other cases as
in a MAR($0,2$) with double root and resonance we can have $h_{0}\neq 0$.

In practice the stochastic maturity $h_N$ has to be predicted. Pointwise
predictions can be computed by considering either the mode of the
distribution of $h_{N}$,  or its expectation $\mathbb{E}[
h_{N}]=h_{0}-\mathbb{E}[N] $. This expectation is not equal to
the mode in general due to the asymmetric causal and noncausal dynamics.
Prediction intervals can also be deduced from the distribution of $N$ . Then these prediction intervals have to be estimated in practice by replacing $\theta$ by $\hat{\theta}_T$ in $h_0$ and in the distribution of $N$, and then  the estimation risk has to be taken into account.

When $p$ and/or $q$ are larger or equal to $2$, the sequence $c_{h}$ can
feature several local maxima, that are turning points, at different levels
due to the multiple roots. This can reveal cluster of bubbles with different rates of explosion. 

\subsection{Serial Dependence in a CEV Framework}

Example 1 of the pure causal AR(1) model shows that the serial dependence (in $h$) in
the CEV framework is very different from the unconditional serial dependence (in $h$)
of $y_{t+h}/|y_{t}|$. Moreover the forward serial dependence (for $h\geq 0$%
) and the backward serial dependence (for $h\leq 0$) can be very different, as seen in the MAR(1,1) case.\vspace{1em}

Let us consider a MAR($1,1$) process: From eq.  (3.6) the tail process $(Z_h)$ can be written:%

\begin{eqnarray*}
Z_{h} &=&\mu ^{-1}\mathbb{I}_{N\leq -h}+\lambda \mathbb{I}_{N\geq -h+1} \\
&=&\lambda +\left( \mu ^{-1}-\lambda \right) \mathbb{I}_{N\leq -h}.
\end{eqnarray*}%
We deduce that:

\begin{eqnarray*}
Cov[ Z_{h},Z_{k}] &=&\left( \mu ^{-1}-\lambda \right)
^{2}Cov[\mathbb{I}_{N\leq -h},\mathbb{I}_{N\leq -k}]\\
&=&\left( \mu ^{-1}-\lambda \right) ^{2}\left[ F_N\left( \min \left(
-h,-k\right) \right) -F_N\left( -h\right) F_N\left( -k\right) \right] ,
\end{eqnarray*}%
where $F_N$ is the c.d.f. of the distribution of $N$. Then the serial
correlation is:

\begin{eqnarray*}
\rho \left( Z_{h}, Z_{k}\right) &=&\frac{Cov[Z_{h},Z_{k}] }{\sqrt{%
\mathbb{V}[Z_{h}]\mathbb{V}[Z_{k}]}} \\
&=&\frac{F_N\left( \min \left( -h,-k\right) \right) -F_N\left( -h\right) F_N\left(
-k\right) }{\left( F_N\left( -h\right) \left( 1-F_N\left( -h\right) \right)
\right) ^{1/2}\left( F_N\left( -k\right) \left( 1-F_N\left( -k\right) \right)
\right) ^{1/2}}.
\end{eqnarray*}

Whereas the underlying process $y_{t+h}/y_{t+h-1}$, $h$ varying, is
stationary, we observe that conditioning by a large value $y_{t}>y$ destroys
the stationarity of the tail process, since $\gamma \left( h,k\right) $ no longer depends on $\left( h,k\right) $ by $h-k$ only. This is largely due to the constraint $X_0=1$, whereas the other $X_h$ values, $h \neq 0$, are not equal to 1.

\section{Alternative Extreme Conditioning}

The standard Conditional Extreme Value Theory (CEVT) applied to time series has mainly considered the conditioning  set $|y_{t}|>|y|$, as in Propositions 1 and 2. However, other conditioning by large values can be considered, such as $%
y_{t}=y $, $y$ large, or $y_{t}>y,y_{t-1}>y$, or $y_{t}>y_{t-1}>y$. In
particular the literature on MAR(1,1) processes with stable distributed errors
has derived closed form expressions of some power moments of $y_{t+h}$ given
$y_{t}=y$, up to power 4 [Fries and Zako\"ian (2019), Fries (2021), Section 3].

It is important to check if a tail process approach can still be used with alternative conditioning sets and to discuss the updating of
predictive distributions with respect to time and with respect to the form
of the extreme conditioning set.

\subsection{Analysis for the Conditioning Set $y_{t}=y$, $y$ large}
Regular variation of the tail of a (multivariate) distribution is an essential tool for describing domains of attraction of multivariate extremes (De Haan and Resnick, 1987, p83).  However, the analysis can be done based on assumptions on either the survival function (corresponding to the conditioning $y_t>y$), or the density function (corresponding to the conditioning $y_t=y$).  The results are often similar,  but may require different assumptions on the (multivariate) distributions.   

In this subsection,  we show that most of the results,  such as Propositions 1 and 2,  remain valid, if the process is positive and we replace $y_t>y $,  $y$ large,  by $y_t=y$,  $y$ large.  The analysis is based on the closed form expression of the transition probability density function (p.d.f.) for a MAR\ process,  as well as the following lemma, which is key to most weak convergence results of the paper.
\subsubsection{The Single Big Jump (SBJ) Principle}

\begin{lm}[Single Big Jump]
	\label{technicallemma}
Assume two independent variables $Z_1$, $Z_2$ whose p.d.f.'s have polynomial (i.e. Paretian) decaying right tails \footnote{Under some mild regularity conditions (such as the Ultimate Monotonicity of the Density (UMD), see the Monotone Density Theorem [Bingham,  Goldie,   Teugels (1989,  Theorem 1.7.2)], 
 the regular variation of the survivor function of $\epsilon$, that is,  $\mathbb{P}[\epsilon>u]=\frac{l_2(u)}{u^{\alpha}}$ for some slowly varying function $l_2(\cdot)$ implies  the regular variation of the corresponding density. } 
$f_{i}(z)=z^{-\alpha-1}l_i(z), i=1,2$,  where $l_1, l_2$ are slowly varying functions at $+\infty$\footnote{That is,  for any given positive $r$,  the ratio $\frac{l_i(rz)}{l_i(z)}$ converges to 1 as $z$ increases to infinity.  }, if moreover they have equivalent p.d.f.'s:
\begin{equation}
		\label{tailequivalence}
		\lim_{z\to +\infty}\frac{f_{1}(z)}{f_{2}(z)}= \xi>0,
		\end{equation} 
 then the conditional distribution of 
 \begin{equation}
 \label{conditionalequality}
 R=\frac{Z_1}{Z_1+Z_2}\mid S=Z_1+Z_2=s
\end{equation} 
  converges weakly (i.e., in distribution) to the Bernoulli distribution with success parameter $p=\frac{\xi}{\xi +1}$,  that is independent of $\alpha$, as $s$ goes to $+\infty$.  
\end{lm}
Note that the conditional distribution of
 \begin{equation}
 \label{conditionalinequality}
 R=\frac{Z_1}{Z_1+Z_2}\mid S=Z_1+Z_2=s,
\end{equation}
 also converges weakly to the same limit as $s$ increases to infinity.  Indeed,  by integrating out the conditional distribution of $S$ conditional on $S>s$,  we get a link between the conditional distributions \eqref{conditionalequality} and \eqref{conditionalinequality}:
$$
\ell(R| S>s)=\int_{s}^\infty \ell(R|S=s) \ell(u|S>s) \mathrm{d} u.
$$
where $\ell(u|S>s)$ is the conditional density of $S$ given that $S>s$.  Thus, if on the right hand side (RHS), $\ell(R|S=s)$ converges to a limit that does not depend on $s$, then on the left hand side $\ell(R| S>s)$ also converges to the same limit\footnote{Assuming that the weak convergence of $\ell(R|S=s)$ is ``uniform". }, but not the other way around. 
 
 
 Lemma \ref{technicallemma} was first proved in the literature by Lehtomaa (2015), under rather restrictive assumptions (i.i.d.  positive variables with log concave density).  In the online  appendix B.1, we provide a more general proof of Lemma 2.  We also provide its extension to more than two variables in Online Appendix B.3.
 
In a time series context,  important examples of random variables with equivalent p.d.f.'s are finite moving averages of i.i.d. errors,  such as $\epsilon_t+ \psi \epsilon_{t+1}$,  with an i.i.d. error term $(\epsilon_t)$ that has Paretian tail.  Then we have [Bingham, Goldie and Omey (2006)]:
$$
\lim_{z \to +\infty } \frac{f_{\epsilon_t+ \psi \epsilon_{t+1}}(z)}{f_{\epsilon}(z)} =
1+\psi^{\alpha}.
$$
The extension of this result to the case of infinite moving average is also possible, but requires some technical conditions. We have:
\begin{lm}
\label{ultimatemonotone}
Let us assume that the density of $\epsilon_t$ is regularly varying at $+\infty$ and that the p.d.f.'s of $\epsilon_t$ and $u_t=\epsilon_t+ \psi \epsilon_{t+1} +\psi^2 \epsilon_{t+2}+\cdots$ satisfy the UMD condition: 
$$
f_u(z)/v_1(z) \rightarrow 1, \qquad f_{\epsilon}(z)/v_2(z) \rightarrow 1,
$$
as $z$ increases to infinity, where $v_1, v_2$ are monotone on a certain interval $(M, \infty)$, 
then we have:
$$
\lim_{z \to +\infty } \frac{f_{u}(z)}{f_{\epsilon}(z)} =
1+\psi^{\alpha}+\psi^{2\alpha}+\cdots =\frac{1}{1-\psi^{\alpha}},  \text{ if } \psi \in (0,1) . 
$$
\end{lm}
\begin{proof}
	See Online Appendix B.2.
\end{proof}
A similar property has been established by Cline (1983) for the ratio of the survival functions of $u_t$ and $\epsilon_t$. 
Intuitively, the UMD condition is needed since a survivor function is ``more regular" than its derivative, i.e. the density.

As an illustration, if $\epsilon_t$ follows Cauchy $\mathcal{C}(0,1)$ with scale parameter $\alpha=1$ and location parameter 0,  then $u_t$ defined in eq. \eqref{ut} is also Cauchy distributed with scale parameter $\frac{1}{1-\psi}$. In this case, the density ratio $\frac{f_u(z)}{f_{\epsilon}(z)}$  has a closed form and we can calculate directly the limiting ratio of the p.d.f.'s: 
$$
\lim_{z\to +\infty}\frac{f_u(z)}{f_{\epsilon}(z)} =\lim_{z\to +\infty}(1-\psi)\frac{1+z^2}{1+(1-\psi)^2z^2}=\frac{1}{1-\psi},
$$
Hence, we recover the formula in Lemma 2 with $\alpha=1$.

\subsubsection{Extremal behaviour of the MAR(1,1)}
Let us first consider the case of a positive stationary MAR$(1,1)$ process to understand the arguments of the proof before discussing the more general framework. We assume that:
\begin{equation*}
\left( 1-\phi L\right) \left( 1-\psi L^{-1}\right) y_{t}=\epsilon _{t},
\end{equation*}%
with $\phi,  \psi \in (0,1)$.\footnote{These restrictions on $\phi$ and $\psi$ are implied by the positivity and stationarity of the process. }

We have both:
\begin{align}
\label{yt} y_t&=\frac{1}{1-\phi \psi} (v_t+\psi u_{t+1}), \\
\label{ytplus1}y_{t+1}&=\frac{1}{1-\phi \psi} (\phi v_t+  u_{t+1}),
\end{align}
where $v_{t+1}=\phi v_t+\epsilon_{t+1}=y_t-\psi y_{t+1}$ is pure causal and $u_t=\psi u_{t+1}+ \epsilon_{t}=y_t-\phi y_{t-1}$ is pure noncausal. 

Because $\frac{1}{1-\phi \psi}v_t$ and $\frac{\psi}{1-\phi \psi} u_{t+1}$ are independent and have equivalent tails,  the limiting ratio of their two p.d.f.'s is: $\xi= \frac{\frac{1}{1-\phi^{\alpha}}}{\frac{\psi^{\alpha}}{1-\psi^{\alpha}}}$.  By applying Lemma 2, $\frac{1}{1-\phi \psi}\frac{v_t}{y_t}$ (resp. $\frac{\psi u_{t+1}}{1-\phi\psi}\frac{1}{y_t}$) converges to the Bernoulli distribution with probability parameter $\frac{\xi}{\xi+1}=\frac{1-\psi^{\alpha}}{1-\phi^{\alpha}\psi^{\alpha}}$ (resp. $\frac{1}{\xi+1}$).   

From eqs. \eqref{yt} and \eqref{ytplus1}, we get:
\begin{equation}
\label{keyprediction}
\frac{y_{t+1}}{y_t} =\phi+\frac{\psi u_{t+1}}{1-\phi\psi}\frac{1}{y_t},
\end{equation}
and we conclude that:
\begin{Proposition}
	\label{equaly}
	For the MAR(1,1) process, the conditional distribution of $r_{t+1}=y_{t+1}/y_{t}$ given $y_t=y$ converges to the discrete variable with masses at $\psi^{-1}$ and $\phi$,  with weights $\frac{1}{\xi+1}=\frac{\psi^{\alpha}-\phi^{\alpha}\psi^{\alpha}}{1-\phi^{\alpha}\psi^{\alpha}}$ and $\frac{\xi}{\xi+1}=\frac{1-\psi^{\alpha}}{1-\phi^{\alpha}\psi^{\alpha}}$,  respectively.   
\end{Proposition}

This result is the analog of Proposition 2 for $h=1$,  with conditioning set $y_t=y$ instead of $| y_t| >y$.  More precisely,  Proposition 2 says that given $|y_t| >y$,  the conditional distribution of $r_{t+1}$ converges to that of $Z_1=c_{N+1}/c_N$,  whose expression is given by eq. (3.6).  In particular,  we have:
$$
\mathbb{P}[Z_h=\lambda]=\mathbb{P}[N \geq 0] = \frac{\frac{1}{1-\lambda^{\alpha}}}{\frac{1}{1-\lambda^{\alpha}}+ \frac{1}{1-\mu^{\alpha}}-1}=\frac{1-\psi^{\alpha}}{1-\phi^{\alpha}\psi^{\alpha}},
$$
where $\psi=\mu, \phi=\lambda$.  

 The fact that we find the same limiting conditional distributions for $Z_1$ with the two different conditioning sets $y_t>y$ and $y_t=y$ is not surprising.  Indeed,  by integrating out $y$,  we get that the conditional distribution of $y_{t+1}/y_{t}$ given $y_t>y$ converges to the same limiting discrete distribution.  
 
 It is also straightforward to extend this result to the joint distribution of $(y_{t+h}/y_t), $ $h$ varying,  given $y_t=y$,  which will converge to the same limiting distribution as in Proposition 2.  For instance,  if we focus on the prediction for the next two periods,  that is,  $(y_{t+1},y_{t+2})/y_t$,  we get (see Online Appendix B.3):
 \begin{Proposition}
 \label{h12}
 	For the MAR(1,1) process, the conditional distribution of $(y_{t+1}/y_t,y_{t+2}/y_t)$ given $y_t=y$ converges to the discrete variable with masses at $(\phi,  \phi^2)$, $(\psi^{-1}, \frac{\phi}{\psi})$ and $(\psi^{-1}, \frac{\phi}{\psi})$,  with weights $\frac{\frac{1}{1-\phi^{\alpha}}}{\frac{1}{1-\phi^{\alpha}}+\frac{\psi^{\alpha}\phi^{\alpha}}{1-\psi^{\alpha}}+1}$,  $\frac{\frac{\psi^{\alpha}\phi^{\alpha}}{1-\psi^{\alpha}}}{\frac{1}{1-\phi^{\alpha}}+\frac{\psi^{\alpha}\phi^{\alpha}}{1-\psi^{\alpha}}+1}$ and $\frac{1}{\frac{1}{1-\phi^{\alpha}}+\frac{\psi^{\alpha}\phi^{\alpha}}{1-\psi^{\alpha}}+1}$,  respectively.   
 \end{Proposition}
Asymptotically,  the two ratios $y_{t+1}/y_t$ and $y_{t+2}/y_{t+1}=\frac{y_{t+2}/y_t}{y_{t+1}/y_t}$ are independent,  with the same limiting distribution with masses at $\phi$ and $\psi^{-1}$,  which is exactly the limiting distribution obtained in Proposition 2.

 \subsubsection{Extremal behaviour of the MAR ($p,q$)}
 
 The general MAR($p,q$) framework can be analyzed in a similar way. We can compute in closed form the transition of the process, and in particular show that this process is Markov of order $p+q$. Then a careful analysis of the behaviour of the distribution of the process conditional on $y_t = y$, with $y$ large, shows a convergence to a tail process with discrete values functions of the moving-average coefficients $c_h$. Let us first recall the expressions of the transition p.d.f.  of a MAR($p,q$) process:
 
 \begin{Proposition}[Fries and Zako\"ian (2019), Proposition 3.1]
\label{markovpq} 
  i)  A MAR$(p,q)$ process is a Markov process of order $p+q$. 
  
 ii).  The transition density of the MAR$(p,q)$ process $(y_t)$ is:
 	\begin{eqnarray*}
 	 l(y_{t+1} | \underline{y_t}) = f_{\epsilon} (\Psi (L^{-1}) \Phi (L) y_{t+1-q}) \times \frac{f_{(u_{t+1}, \ldots, u_{t+1-q})} (\Phi (L) y_{t+1}, \ldots, \Phi (L) y_{t+2-q})}{f_{(u_{t},\ldots, u_{t+1-q})} (\Phi (L) y_t, \ldots, \Phi (L) y_{t+1-q})},
 	\end{eqnarray*}
 where $f_{(u_t,\ldots, u_{t+1-q})}$ denotes the joint density of the pure causal component.
 	\end{Proposition}
This last formula says that the conditional distribution of $ l(y_{t+1} | \underline{y_t})$ is equal to the conditional distribution of the noncausal process $(u_t)=(\Phi (L) y_{t})$ given its own past: $ l(y_{t+1} | \underline{y_t}) = f_{\epsilon} (\Psi (L^{-1}) u_{t+1-q}) \times \frac{f_{(u_{t+1}, \ldots, u_{t+1-q})} (u_{t+1}, \ldots, u_{t+2-q})}{f_{(u_{t},\ldots, u_{t+1-q})} (u_t, \ldots, u_{t+1-q})}$.  In other words, it is often convenient to transform $(y_t)$ to $(u_t)$ and analyze the asymptotic behaviour of this latter process first,  which is Markov of order $q$. 
 
 This closed form expression of the transition density can be used to derive the conditional distribution of $(y_{t+h}/y_t, h$ varying) given $y_t$  and its behaviour when $y_t = y$ is large (see Appendix A.3).
 
 \begin{Proposition} For a MAR$(p,q)$ process, the conditional distribution of $(y_{t+h}/y,h$ varying) given $y_t = y$ tends to a discrete distribution with support $(c_{h+n}/c_n, h$ varying) with $ n \in \mathbb{Z}$, when $y$ tends to infinity.\end{Proposition}
 This proposition is the analog of Proposition 2.

\subsection{Analysis for the Conditioning Set $y_{t}=y$ and $r_t,  r_{t-1},...$,  with $y$ large}
When the conditioning set $y_t>y$ is replaced by $y_t=y$, we get a similar limiting distribution by Proposition 7. The aim of this subsection is to show that when extra information is incorporated into the conditioning set,  the limiting conditional distribution could be updated to reflect this extra information.

\subsubsection{Extremal behaviour of MAR(1,1)}
Let us first consider the case of a MAR(1,1) process.  From Proposition \ref{markovpq},  the conditional p.d.f.  of $y_{t+1}$ given $y_{t},y_{t-1},y_{t-2},...$ is
equal to:%
\begin{eqnarray*}
\ell \left( y_{t+1}|y_{t},y_{t-1},y_{t-2},...\right) &=&\frac{f_{\left(
u_{t+1},u_{t}\right) }\left( y_{t+1}-\phi y_{t},y_{t}-\phi y_{t-1}\right) }{%
f_{u_{t}}\left( y_{t}-\phi y_{t-1}\right) } \\
&=&\frac{f_{u_{t}|u_{t+1}}\left( y_{t}-\phi y_{t-1}|y_{t+1}-\phi
y_{t}\right) f_{u_{t+1}}\left( y_{t+1}-\phi y_{t}\right) }{f_{u_{t}}\left(
y_{t}-\phi y_{t-1}\right) } \\
&=&\frac{f_{u}\left( y_{t+1}-\phi y_{t}\right) }{f_{u}\left( y_{t}-\phi
y_{t-1}\right) }f_{\epsilon }\left( y_{t}-\phi y_{t-1}-\psi \left(
y_{t+1}-\phi y_{t}\right) \right),
\end{eqnarray*}

\noindent where $f_{\epsilon}$ and $f_u$ are the p.d.f.'s of the error $\epsilon_t$ and the pure noncausal component $u_t$, respectively.  In particular,  the MAR(1,1) process is second-order Markov and,  it suffices to condition on $y_t$ and $r_t$ to get its predictive distribution given all the past.

Then using the same type of proof as in Proposition \ref{equaly},  we get: 

\begin{Proposition}
	\label{yr}
	 For the MAR (1,1) process, the conditional
distribution of $r_{t+1}$ given $y_{t}=y,r_{t}=r$, with $y$ large and $r$ fixed, converges (weakly) to a
discrete distribution with masses at $\phi $ and $\phi +\psi ^{-1}\left(
1-\phi /r\right) $,  with weights $1-\psi^{\alpha}$ and $ \psi^{\alpha}$,  when $y$ tends to infinity.\end{Proposition}

\textbf{Example:  Cauchy MAR(1,1).} 
As an illustration, let us consider the MAR(1,1) Cauchy process, with $\epsilon_t\sim \mathcal{C}(0,1)$. In this case, the conditional density of $r_{t+1}$ given $y_{t}=y,r_{t}=r$ is:
\begin{equation}
	\label{thedensity}
f(r_{t+1}|y_t=y,r_t=r)= \frac{\pi |y|[1+y^2(1-\phi/r)^2(1-\psi)^2]}{[1+y^2(r_{t+1}-\phi)^2(1-\psi)^2][1+y^2(1-\phi/r-\psi r_{t+1}+\phi \psi )^2(1-\psi)^2]}
\end{equation}
We can distinguish two cases:
\begin{itemize}
	\item When $r_t$ is close to $\phi$, the numerator and the second factor of the denominator partially cancel out and this pdf is close to $(1-\psi)  \frac{(1-\psi) |y|}{\pi  [(1-\psi)^2 y^2 (r-\phi)^2+1]}$, where we note that the ratio $ \frac{(1-\psi) |y|}{\pi  [(1-\psi)^2 y^2 (r-\phi)^2+1]}$ is the density of $\frac{Z_1}{(1-\psi)y}$, where $Z_1$ follows $\mathcal{C}(\phi,1)$. Clearly, $\frac{Z_1}{(1-\psi)y}$ converges weakly to the point mass at $\phi$ as $y$ increases to infinity. This latter pdf is often called the Lorentzian function. 
	\item When $r_t$ is close to $\phi +\psi ^{-1}\left(
	1-\phi /r\right) $, the numerator and the first factor of the denominator partially cancel out, and the RHS of eq. \eqref{thedensity} is close to $\psi \frac{\psi |y|}{\pi [1 + \psi^2 (-\phi + r - (1 - \phi/r)/\psi)^2 y^2]}$, with the ratio $\frac{\psi |y|}{\pi [1 + \psi^2 (-\phi + r - (1 - \phi/r)/\psi)^2 y^2]}$ representing the density of $\frac{Z_2}{\psi y}$, where $Z_2$ follows $\mathcal{C}(\phi +\psi ^{-1}\left(
	1-\phi /r\right),1).$ Clearly, $\frac{Z_2}{\psi y}$ converges weakly to the point mass at $\phi +\psi ^{-1}\left(
	1-\phi /r\right)$ as $y$ increases to infinity. 
\end{itemize}
Thus, we have:
\begin{align*}
f(r_{t+1}|y_t=y,r_t=r)&= (1-\psi) \Big[  \frac{(1-\psi) |y|}{\pi  [(1-\psi)^2 y^2 (r-\phi)^2+1]}\Big]+ \psi \Big[ \frac{\psi |y|}{\pi [1 + \psi^2 (-\phi + r - (1 - \phi/r)/\psi)^2 y^2]} \Big]\\
&+e(r_{t+1},y,r)
\end{align*}
where $e(r_{t+1},y,r)$ is the approximation error.  The sum of the first two terms converges towards a weighted average of two point masses, whereas using tedious elementary algebra, we can show that the error term is such that $\int e(r_{t+1},y,r) g(r_{t+1}) \mathrm{d}r_{t+1}=O(1/y)$ for any bounded function $g$.  This analysis shows that the convergence of $f(r_{t+1}|y_t=y,r_t=r)$ towards a mixture of point masses is at the rate of $1/y$. 
\vspace{1em}

We plot,  in Figures 1 and 2, the predictive distribution of $r_{t+1}$ given $y_t$ and $r_t$.  We take $\phi=0.6$ and $\psi=0.4$ in the MAR(1,1) Cauchy process,  and consider $y_t=100$ in both figures.  We set $r_t=2$ in Figure 1,  and $r_t=1$ in Figure 2.  As expected,  we get two modes for the conditional distribution in both cases.  The heights of the two modes are roughly similar across both figures,  but the location of one of the two modes differs between the two figures.  This is consistent with Proposition 8,  which says that the limiting distribution depends on $r_t$ only through the location of $\phi +\psi ^{-1}\left(
1-\phi /r_t\right) $.
\begin{figure}[H]
\centering
\includegraphics[scale=0.4]{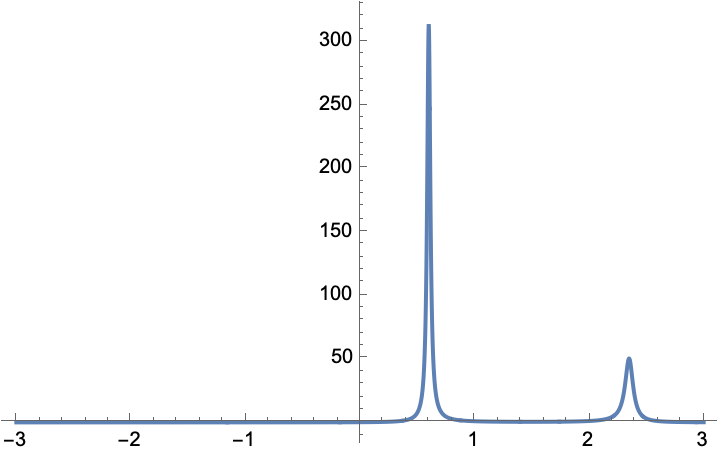}
\caption{Predictive Distribution of $r_{t+1}$ given $y_t=100$ and $r_t=2$ in a MAR(1,1) Cauchy model.}
\end{figure}

\begin{figure}[H]
\centering
\includegraphics[scale=0.4]{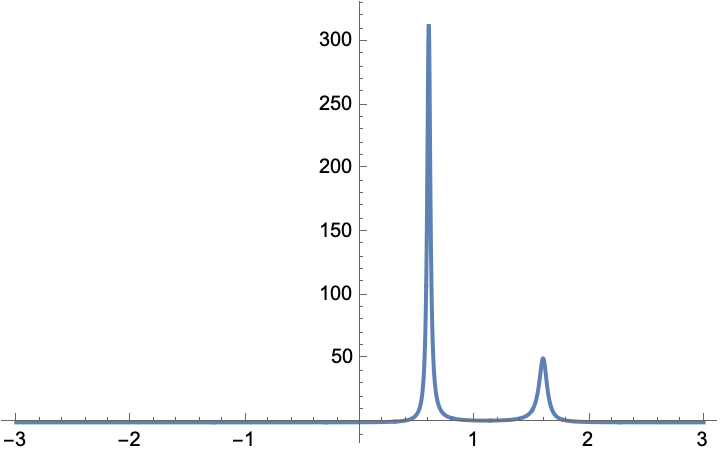}
\caption{Predictive Distribution of $r_{t+1}$ given $y_t=100$ and $r_t=1$ in a MAR(1,1) Cauchy model.}
\end{figure}


To our knowledge,  this kind of result in Proposition 8,  where the limiting distribution of the normalized future path $y_{t+h}/y_t$,  $h=1,2,\cdots$  depends on $r_t$,  has never been established before in the noncausal literature.   It can be useful in practice,  since even though Propositions 2 and \ref{equaly} say that,  conditional on $y_t$ being large,  the distribution of the past normalized values $y_{t+h}/y_t$,  where $h=-1,-2,\cdots$ converges to a discrete distribution,  their realized values are almost surely different from these limiting values.  Then Proposition \ref{yr} explains how to reconcile these realized values with Proposition \ref{equaly}.  

Let us now discuss the main differences between Propositions \ref{yr} and \ref{equaly}. 

\paragraph{ $i)$ Case 1:  $r_t$ takes  limiting values. } In Proposition \ref{yr}, we get potentially different locations for the two limiting point masses compared to Proposition \ref{equaly},  as well as new weights for the two point masses.  Moreover,  the new location of the second mass now depends on the value of $r_t$.  However, Proposition \ref{equaly} also says that this value is expected to be concentrated around two points $\phi$ and $\psi^{-1}$.  When $r_t$ takes exactly one of these two limiting values,  the locations of the two masses in Proposition \ref{yr} agree with those found in Proposition \ref{equaly},  since $\phi +\psi ^{-1}\left(1-\phi /\psi^{-1}\right) =\psi^{-1}$.   

Nevertheless,  even in the case where $r_t$ takes one of the most probably values $\phi$ and $\psi^{-1}$,  the weights of the two point masses are different in Propositions \ref{yr} and \ref{equaly}.  Thus the inclusion of $r_t$ in the conditioning set has an effect of (Bayesian) updating the weights of the two limiting point masses.  More precisely,   from the convergence of the conditional distribution of $l(\frac{y_{t-1}}{y_t}, \frac{y_{t+1}}{y_t})$ given $y_t=y$ towards $(X_{-1}, X_1)$,  we deduce that the conditional distributional of $\frac{y_{t+1}}{y_t}$ given $y_t=y,  \frac{y_{t}}{y_{t-1}}=r$ converges to the distribution of $X_1$ given $X_{-1}=X_0r^{-1}$,  where $r$ is either $\phi$,  or $\psi^{-1}$.  Since $X_{-1}=X_0(c_{N}/c_{N-1})^{-1}$,  whose expression is given in eq. \eqref{mar110},   the knowledge of the (limiting) value of $r_t$ provides information on the stochastic drift $N$,  and then on the turning point.  Indeed, we get:
\begin{align*}
r_t= \phi & \Longleftrightarrow N \geq 1, \\
r_t= \psi^{-1} & \Longleftrightarrow N \leq 0.
\end{align*}
In the first case,  we have $X_{1}=X_0c_{N+1}/c_N=\phi$,  since $N \geq 1$ implies \textit{a fortiori} that $N+1 \geq 1$,  and this value $\phi$ is consistent with Proposition \ref{yr},  since the location of the second point mass $\phi+ \psi^{-1}(1-\phi/r)$ is equal to the location of the first point mass $\phi$.  

In the second case,  $X_{1}=X_0c_{N+1}/c_N$ can be equal to either $\phi$ or $\psi^{-1}$,  with $$\mathbb{P}[X_1=\phi]=\mathbb{P}[N+1 \geq 1 \mid N \leq 0]=\frac{\mathbb{P}[N = 0]}{\mathbb{P}[N \leq 0]}=\frac{1-\phi^{\alpha}}{1-\phi^{\alpha}\psi^{\alpha}}.$$
Thus,  we find the same probability as in Proposition \ref{yr}.   Proposition \ref{yr} is an improvement over Proposition \ref{equaly},  which $i)$ discusses the case where $r_t$ might not be equal to one of the two limiting values; $ii)$ conducts a Bayesian updating on the distribution of $N$ using the additional information on $r_t$.

\paragraph{$ii)$ 
	Case 2:  $r_t$ does not take limiting values. } How to interpret the general case,  when $r_t$ is not necessarily equal to $\phi$ or $\psi^{-1}$? In the MAR(1,1) case,  the causal-noncausal decomposition \eqref{pqvu} becomes: 
$$
y_t=\frac{1}{1-\psi \phi}(\phi v_{t-1}+u_t).
$$
Given $y_t$ and $r_t=y_{t}/y_{t-1}$ with $r_t$ different from $\phi$ and $\psi^{-1}$, we have the following properties: 

\begin{lm}
\begin{itemize}
\item In the decomposition $y_t=\frac{1}{1-\psi \phi}(\phi v_{t-1}+u_t)$, both $v_{t-1}=y_{t-1}-\psi y_t$ and $u_t=y_t-\phi y_{t-1}$ are known,  and are both large,  in the sense that,  for fixed $r_t=r$,  both $v_{t-1}$ and $u_t$ go to infinity when $y_t$ increases to infinity.  
\item  The relationship between $(y_t, r_t)$ and $(v_{t-1}, u_t)$ is one-to-one.  This implies that the information set $(y_t, r_t)$ is equivalent to the information set $(v_{t-1}, u_t)$. 
\item Moreover, $v_{t-1}$ and $u_t$ are mutually independent.  In particular,  the conditional distribution of $(u_{t+h})_h /u_t$ given $(y_t, r_t)$,  or equivalently given $(v_{t-1}, u_t)$,  is equal to the conditional distribution of $(u_{t+h})_h /u_t$  given $u_t$ only,  with $u_t$ large. 
\end{itemize}
\end{lm}
The proof of Lemma 4  is obvious and omitted.  It suggests that we can apply Lemma 2 to $v_{t-1}$ and $u_t$ separately,  or equivalently, to investigate the two one-sided tail processes $(X_{1,h})$ and $(X_{2,h})$ associated with these two processes.  In particular,  since $(u_t)$ is pure noncausal,  given $u_t$ is large,  the conditional distribution of $u_{t+1}/u_t$ converges to a Bernoulli distribution with point masses at 0 and $\psi^{-1}$,  with weights $1-\psi^{\alpha}$ and $\psi$.  Finally,  we have:
\begin{equation}
\label{predictytplus1}
y_{t+1}/y_t= \frac{1}{1-\phi \psi}(\phi^2 v_{t-1}+\phi \epsilon_t + u_{t+1})=  \frac{1}{1-\phi \psi}
  \Big(\phi^2 \frac{v_{t-1}}{y_t}+\phi \frac{u_t}{y_t} + (1-\phi \psi) \frac{u_{t+1}}{u_t}\frac{u_t}{y_t}\Big),
\end{equation}
where the ratios $v_{t-1}/y_t$ and $\frac{u_t}{y_t}$ are known.   Thus the limiting distribution of $y_{t+1}/y_t$ is obtained from the limiting distribution of $u_{t+1}/u_t$ through given linear location and scale changes.  Formula \eqref{predictytplus1} highlights the critical role played by the pure noncausal process in decomposition \eqref{pqvu} (equal to $ \frac{1}{1-\phi \psi} u_t$ in the MAR(1,1) case).  In fact, the conditioning with respect to all past values of the process $y_t$ has the effect of conditioning out the causal part $\frac{1}{1-\phi \psi} v_{t-1}$ and it suffices to consider the prediction of the noncausal part. 

This result can also be interpreted using the one-sided tail processes $(X_{1,h})$ and $(X_{2,h})$.  Indeed,  in Section 3.3 we have seen that these two processes are independent,  and for prediction purpose ($h >0$),  $X_{1,h},  h>0$,  is deterministic.  It suffices to predict the one-sided tail process $(X_{2,h})$ associated with the noncausal part of $y_t$.

\paragraph{$iii)$ Single vs Double Big Jumps.} By applying twice the SBJ principle to $u_t=\sum_{h=0}^\infty \psi^h \epsilon_{t+h}$ and $v_{t-1}=\sum_{h=0}^{\infty} \phi^{h} \epsilon_{t-1-h}$,  respectively.  We see that there are two big jumps among the sequence $(\epsilon_t), t\in \mathbb{Z}$,  one with negative index (corresponding to a past date) and the other one with nonnegative index (corresponding to a future date or the current date). \footnote{In the special cases where $r_t=\phi$ (resp.  $r_t=\psi^{-1}$),  we have $u_t=0$,  or $v_{t-1}=0$.  In this case, the DBJ becomes SBJ. } In other words,  with the conditioning set  $y_{t}=y,r_{t}=r$,  instead of having a Single Big Jump (SBJ),  we are in a new framework of ``Double Big Jumps" (DBJ)\footnote{In extreme value theory jargon, this  can be interpreted as a kind of hidden regular variation [see e.g.  Resnick and Roy (2014)],  as opposed to standard regular variation,  which underlies Proposition 1 and the main theories in Kulik and Soulier
 (2020). } 
In the special case of MAR(1,1) processes,  this DBJ is particularly simple to analyze,  since to predict $y_{t+1}$,  by eq. \eqref{predictytplus1},  we do not need to infer the location of the past SBJ in $v_{t-1}$, but only the future SBJ in $u_t=\epsilon_t+ \psi u_{t+1}$. 

\subsubsection{Extremal behaviour of MAR($p, q$)}
A MAR($p, q$) process is Markov of order $p+q$,  and thus we are interested in the predictive distribution of $y_{t+1},y_{t+2},...$ given 
$y_t$, $r_t, r_{t-1},...,r_{t-p-q+2}$.  As in Lemma 3,  the idea is to rely on the causal-noncausal decomposition of process $(y_t)$ to transform the analysis of $(y_{t+h})/y_t$ into the analysis of its noncausal part.  More precisely,  in eq. \eqref{pqvu},  the first term 
$
p_t= L^q b_1(L) v_t= L^q b_1(L) \Psi(L^{-1}) y_t 
$
only depends on the current and past $p+q$ observations of $(y_t)$,  and are known given $y_t$, $r_t, r_{t-1},...,r_{t-p-q+2}$.  The second term $f_t=b_2(L)u_t$ is also known.  Thus,  as in Lemma 3, we can analyze the asymptotic distribution of $f_t$ given its own past.   

Moreover,  in the special case where $q=1$,  the degree of $b_2(L)$ is zero.  Thus it is a constant,  which is nonzero.  In this case, $f_t$ is a MAR($0, 1$) process.  Thus,  as in the MAR$(1,1)$ case,  we can apply the SBJ to the pure noncausal AR(1) process $(u_t)$.  
This is an important special case,  since empirical studies often find $q=1$ for economic data,  corresponding to a single noncausal root [see e.g.  Hecq,  Velasquez-Gaviria (2025)].   For such a process,  where $u_t=\Phi(L)y_t$ is a noncausal AR(1), we obtain the following analog of Proposition 8:
\begin{Proposition}
In the MAR($p,1$) model,  as $y$ increases to infinity,  the conditional distribution of $r_{t+1}=y_{t+1}/y_t$ given $y_t=y$ and $r_{t},r_{t-2},...,r_{t-p+1}$ converges to a discrete distribution with two point masses.  The locations of the two point masses are obtained by solving $\Phi(L)y_{t+1}=0$ and $\Phi(L)y_{t+1}=\psi^{-1}\Phi(L)y_{t}$ for $r_{t+1}$,  respectively\footnote{For instance,  solving $\Phi(L)y_{t+1}=0$ leads to:
$r_{t+1}=\phi_1+\frac{\phi_2}{r_{t}}+\frac{\phi_3}{r_{t}r_{t-1}}+\cdots+ \frac{\phi_p}{r_tr_{t-1}\cdots r_{t-1+p}}$.  }. In particular, the locations only depend on $r_{t},r_{t-2},...,r_{t-p+1}$, but not on $y$, and the weights of these two point masses are $1-\psi^{\alpha}$ and $ \psi^{\alpha}$, respectively.  
\end{Proposition}
The formal proof of Proposition 9 is similar to the proof of Proposition \ref{yr} and is omitted.  

For a general MAR($p, q$) processes with $q \geq 2$,  one can transform the problem of predicting $y_{t+1}$ given its own past into predicting $u_{t+1}$ by its own past,  with $(u_t)$ being a pure noncausal AR($q$) process. A systematic treatment of these higher order noncausal processes is out of scope of this paper [see e.g.  Rhee,  Blanchet and Zwart (2019), Dombry, Tillier and Winterberger (2022)].  However,  we conduct an informal discussion of this problem in Appendix B.5.

\subsection{Online Updating : the Conditioning Set $y_{T+1}> y_T>y$,  $y$ large}

Let us now consider an online framework, where at date $T$ the observation is such that $y_T > y$,  where $y$ is large. Therefore we can apply the CEV framework. Let us assume that a new observation is available at date $T+1$ and that we observe that $y_{T+1} > y_T$. Therefore we are possibly in an increasing episode of the bubble and this additional information will change our view on the turning point. By Proposition 1, we have:
$$
\mathcal{L}(y_{T+1}/y_T, \ldots, y_{T+h} / y_T |y_T > y) \stackrel{d}{\longrightarrow} \mathcal{L} (X^{(T)}_{1},\ldots, X^{(T)}_{h}),
$$

\noindent where the tail process $(X^{(T)}_h)$ is indexed by the exogenous date $T$.\footnote{It is important to write explicitly the index $T$ to highlight the difference with the myopic updating considered in Section 4.4.} We deduce that:
\begin{align}
&\mathcal{L} (y_{T+1} / y_T, \ldots, y_{T+h} / y_T| y_{T+1} > y_T > y) \nonumber \\  
=& \mathcal{L} (y_{T+1} / y_T, \ldots, y_{T+h} /y_T|y_T > y, y_{T+1}/y_T>1) \nonumber \\  
&\stackrel{d}{\longrightarrow} \mathcal{L} (X^{(T)}_{1}, \ldots, X^{(T)}_{h}| X^{(T)}_{1} > 1).\label{lessgood}
\end{align}

This limiting distribution can be written under closed form by applying Propositions 1 or 2. For expository purpose let us assume that the moving-average coefficients are nonnegative and apply Proposition 2. We have:
$$
Z^{(T)}_{h} = c_{N^{(T)} + h} /c_{N^{(T)}+ h-1},
$$
where the stochastic drift $N^{(T)}$ driving the tail process also depends on $T$. 

The conditioning event $X^{(T)}_{1} > 1$ is equivalent to $c_{N^{(T)} +1} > c_{N^{(T)}}$, or to an event written on the stochastic drift: $N_T \in A := \{ j : c_{j+1} > c_j\}$. Therefore the link between the stochastic drift and the tail process is unchanged, but the distribution of $N^{(T)}$ has to be updated. \vspace{1em}

The usefulness of such online updating depends on the model.  For MAR(1,1) and MAR($p,1$),  such an updating is less accurate than an updating using the full past information as the conditioning set (see Propositions 8 and 9). For these processes,  the limiting distribution given all past information has a simple form and the latter updating can be applied.  This is not the case  for more general MAR($p, q$) processes with $q \geq 2$. Then the updating in eq. \eqref{lessgood} allows to account for more information than just the current observation,  and hence provides more accurate prediction than a straightforward application of Proposition 1.  

As a toy example,  let us consider again the MAR(1,1) process 
(see Subsection 3.1.4). With positive parameters $\lambda, \mu$, we get:

$$
Z^{(T)}_{1} > 1 \;\mbox{iff}\; N^{(T)} \leq -1.
$$

Then the initial extreme distribution of $N^{(T)}$ in Subsection 3.2.4 is updated to a single geometric distribution in reversed time that weights only the values $j \leq -1.$ The interpretation is that since $T-N^{(T)}$ indicates the time index of the single big jump,  the fact that $Y_{T+1}>Y_T$ suggests that this big jump occurs posterior to date $T$.  Thus,  while under the conditioning set $y_T>y$ only,  the bubble can either collapse or further accumulate at date $T+1$,  under the updated conditioning set $y_{T+1}> y_T>y$,  the bubble can only further accumulate between $T$ and $T+1$.  

\vspace{1em}

The analysis above can be easily extended to an additional observation available at $T+2$ with $y_{T+2} > y_{T+1}$. Then we know that we are in an increasing episode of the bubble since two periods, and so on.

\subsection{Comparison with Myopic Updating}
Let us consider the same setting as in Section 4.2,  and compare the exogenous  limiting distribution obtained in eq. \eqref{lessgood} with a naive application of Proposition 2, either at exogenous date $T$,  or $T+1$.  In the first case,  we use the distribution of the tail process:
\begin{equation*}
X^{(T)}_{h}=c_{N^{(T)}+h}/c_{N^{(T)}}, h \geq 1,
\end{equation*}
to approximate the distribution of $y_{T+1}/y_{T},y_{T+2}/y_{T}, ...$ given $y_{T}>y$.  Note that the tail process is indexed by $T$. Then we use $%
y_{T+h+1}/y_{T+1}=\left( y_{T+h+1}/y_{T}\right) /\left( y_{T+1}/y_{T}\right)
$ to deduce the distribution of $y_{T+2}/y_{T+1},y_{T+3}/y_{T+1}, ...$

In the second case, we use the new tail process $X^{(T+1)}_{h}=c_{N^{(T+1)}+h}/c_{N^{(T+1)}}$, indexed by the new date $T+1$,  to approximate $y_{T+2}/y_{T+1}, y_{T+3}/y_{T+1}, ...$

As a consequence,  the two predictive distributions:
\begin{equation*}
\ell \left( \frac{y_{T+2}}{y_{T+1}},...,\frac{y_{T+h}}{y_{T+1}}%
|y_{T}>y\right) \quad \text{and}\quad \ell \left( \frac{y_{T+2}}{y_{T+1}}%
,...,\frac{y_{T+h}}{y_{T+1}}|y_{T+1}>y\right)
\end{equation*}%
are approximated by the two tail distributions:%
\begin{equation}
\label{l12}
\ell\left( \frac{X^{(T)}_{2}}{X^{(T)}_{1}},...,\frac{X^{(T)}_{h}}{X^{(T)}_{1}}\right)  \quad \text{and}\quad \ell \left( X^{(T+1)}_{1},...,X^{(T+1)}_{h-1}\right).
\end{equation}

We easily check that $(X^{(T+1)}_{1}, \ldots, X^{(T+1)}_{h-1})$ and $(X^{(T)}_{2}/X^{(T)}_{1}, \ldots, X^{(T)}_{h}/X^{(T)}_{1})$ are deduced from a same deterministic transformation of the stochastic drift variables $N^{(T+1)}$ and $ 1+N^{(T)}$, respectively.  These two variables have the same distributions, as both are interpreted as the time index of the Single Big Jump given $Y_{T+1}>y$ and $Y_T>y$,  respectively.  Thus the two distributions in eq.  \eqref{l12} are equal,  and clearly,  both are less accurate than the distribution in eq. \eqref{lessgood},  which uses more conditioning information. 
 



\vspace{1em}

\begin{Remark}The results are derived assuming that the date $T$ is ``exogenous". Therefore, they are not valid if for instance $T$ is the first observed exceedance  date, since the conditioning set would be $y_T > y$ and $y_t <y, t=1,\ldots, T-1$ (see Section 4.5 below).
\end{Remark}
\subsection{CEV for the First Large Exceedance}

The CEV framework in Proposition 1 assumes an exogenous date $t$ and standardizes the data by the value $y_t$ in the positive case, by $|y_t|$ in the general case. It is also possible to perform a CEV analysis, when the date $t$ is the first large exceedance date, such that $y_t > y$, $y_{t-h} < y, \forall  h > 0$, and with $y$ large.

In this framework, the date becomes endogenous (this is a stopping time $\tau$ of the history of the process and then of the tail process $(X_h)$.  This induces a change of the conditional p.d.f. that has to account for the evolution of the binary process $\mathbbm{1}_{y_{t-h}<y_t}, h>0$.  This change has an effect on the analysis of the single big jump occurrence, and more specifically on the asymptotic behaviour of the distribution $\mathbb{P}[y_1+\cdots+y_{\tau}>s] \sim \mathbb{P}[\max_{t \leq \tau} y_{\tau} >s]$ for large $s$ [see Holl and Barkai (2021), eqs.  (21)-(22)].  It also has an effect when analyzing the tail process.  

Such results have been derived in Basrak and Segers (2009), Planinic and Soulier (2018) (PS (2018)). We follow below their notations and consider a positive process for expository purpose. 
 Then we get a weak convergence to another tail process in $h$ denoted $(Q_h)$. This result is valid under a condition [PS (2018), condition (3.1)], satisfied for mixed causal-noncausal processes. Moreover the distribution of the tail process $(Q_h)$ is linked to the distribution of  the tail process $(X_h)$.\vspace{1em}

\begin{Proposition}[PS (2018), Th 3.1 and Lemma 3.7] The distribution of process $(Q_h)$ is obtained from the distribution of process $(X_h)$  by a transformation and a change of probability proportional to $1/[\sum_{h \in \mathbb{Z}} X_h^{\alpha}]^{1/\alpha}.$ More precisely:

\noindent for any shift invariant and homogeneous function $A$ defined on $\mathbb{R}^Z$, we have:

$$
\vartheta E [A (Q)] = E [A(X)/ \|X\|^{\alpha}_{\alpha}],
$$
\noindent where $\vartheta  $ is a constant whose value depends on the dynamics of the process, and  $ \|X\|^{\alpha}_{\alpha}=\sum_{h \in \mathbb{Z}}|X_h|^{\alpha}$. 
\end{Proposition}


Since the tail process $(X_h)$ has a countable support of possible trajectories, the tail process $(Q_h)$ has the same support. Then their distributions differ by the distribution of the stochastic drift $N$.\vspace{1em}

Proposition 10 shows that we move from the distribution of $(X_h)$ to the distribution of $(Q_h)$ by a change of probability measure given by $1/(\vartheta \|X\|^{\alpha}_{\alpha})=(1/  \|X\|^{\alpha}_{\alpha})/E [1/ \|X\|^{\alpha}_{\alpha}]$. \footnote{In the MAR framework, we know from the deterministic recursive equations in Proposition 3 that the tail process $(X_h)$ tends to zero at $\pm \infty$. Thus it satisfies the so-called anticlustering condition in Basrak and Sergers (2009, Proposition 4.2).  In this framework, $\vartheta=\mathbb{P}[\sup_{h \leq -1} X_h \leq 1]$.} This change is due to the new conditioning that involves a density ratio of the type:
$$
\displaystyle \frac{\prod_{k=1}^t  \mathbbm{1}_{y_{t -k}<y_t}}{\mathbb{P}[y_{t-k}<y_t,  k=1,...,t]}=\prod_{k=1}^t \Big\{  \mathbbm{1}_{y_{t -k}<y_t/ \mathbb{P}[y_{t-k}<y_t \mid  y_0<y_t, \cdots, y_{t-1}<y_t]} \Big\},
$$
in which the last conditioning $y_t>y$ is not introduced for expository purpose.  This change of probability measure has no effect on the support of the distributions of $(Q_h)$ and $(X_h)$, in particular on the deterministic relationships satisfied by their components.  In other words,  Proposition 1 and Corollary 1 are also valid for the tail process associated with this first exceedance conditioning,  
that are $(Q_h)$,  $(U_h^*)$ and  $(V_h^*)$.  For instance we have:
\begin{align*}
U_h^*=0,  & \text{ if } h \geq 1-N^{*}, \\ 
V_h^*=0,  & \text{ if } h \leq -1-N^{*}, 
\end{align*}
where $N^{*}$ is the stochastic drift associated with $(Q_h)$. 


To summarize the results of Section 4, in the framework of MAR processes,  CEV theory can lead to different tail processes for which the distributions have sometimes a same support, but the law of the stochastic drift depends on the normalization and on the conditioning set. These differences can lend to fallacies and pitfalls when interpreting the results [see Dress and Janssen (2017) for a discussion].

\section{Pure Residual Plots}
Let us now explain how the previous asymptotic results on conditional extreme behaviours can be used to construct diagnostic tools to detect and analyze bubble episodes. In an online perspective of this use, we assume that the parameters in the causal/noncausal polynomials have been estimated from a set of observations $y_1,...,y_T$, and fixed at their estimated values later on, that is at date $t=T+1,T+2,...$ Then the diagnostic tools will be figures, plots,  or scalars computed in a sequential way at each exogenous time $t$ when a new observation arises. The detection will be based on the expected evolution of these tools with $t$. Different tools can be proposed.
\begin{enumerate}[$i)$]
	\item Gouri\'eroux, Hencic and Jasiak (2021), Section 4.5, propose to follow over time $t$ the isodensity curves of the conditional distribution of $(y_{t+1},y_{t+2})$ given $y_t$. This approach is difficult to extend to a horizon larger than 2 and such figures are not standard for practitioners. 
	\item Other authors [see e.g. Phillips, Shi and Yu (2015a, b), Giancaterini et al. (2025)] propose to follow scalar test statistics of a ``bubble hypothesis". However, by summarizing the tool to a scalar, a lot of information could be lost.\footnote{As usual, when a residual plot is replaced by a scalar portmanteau statistic.}
\end{enumerate}
In this section, we consider more standard diagnostic tools, that are conditional residual plots for the pure causal and noncausal components. 
\subsection{The problem}
In MAR processes, different types of errors are involved, that are the i.i.d. errors $\epsilon_t$ appearing in the strict moving average representation, the pure noncausal and causal components $u_t$, $v_t$, respectively, and the pure causal innovations $\eta_t$, say [see Section 4, Gouri\'eroux and Jasiak (2024)]. They all depend on observations and true value of parameters in $\Phi$ and $\Psi$, and can be approximated by replacing the parameters by consistent, asymptotically normal estimators computed on all observations $y_1,...,y_T$. Thus we can construct different types of residuals: $\hat{\epsilon}_{t,T}$, $\hat{u}_{t,T}$, $\hat{v}_{t,T}$, $\hat{\eta}_{t,T}$, say, that are doubly indexed by $t$ and $T$, and define different triangular arrays. The $\hat{\epsilon}_{t,T}$ and the $\hat{\eta}_{t,T}$, $t=1,...,T$, have been used in the literature to define several specification tests of the MAR hypothesis, generally based on different portmanteau statistics, each of them being computed on all the observations [see e.g. Fries and Zako\"ian (2019)]. Their asymptotic behaviour is analyzed, assuming that the number of observations $T$ increases to infinity. Thus these approaches do not distinguish the dates of extreme risks from the other dates. 

The pure residuals $\hat{u}_{t,T} $, $\hat{v}_{t,T} $ can be used in a different way, that is date by date, in a double asymptotic framework, when $T$ increases to infinity and a future date $t=T+1,...,T+h-1$ such that $y_t>y$, with $y$ large, \textit{i.e.}, $y$ tending to infinity. Let us now explain the asymptotic behaviour of $\hat{u}_{t,T}$, $\hat{v}_{t,T} $ and how these new conditional residual plots can be used. 

For expository purpose, we focus on a positive process $(y_t)$ and on pure normalized noncausal residuals at a date $t$ when an extreme observation may arise. 
\subsection{Pure noncausal residual plots}
The pure noncausal residuals at date $t$ are given by:
\begin{equation}
	\label{purenoncausal}
	\hat{u}_{t,T} =\hat{\Phi}_T(L)y_{t}=\Big(1,-\hat{\phi}_{1,T},\cdots, -\hat{\phi}_{p,T} \Big) \Big(y_{t},y_{t-1},\cdots, y_{t-p}\Big)',
\end{equation}
where the parameters are replaced by their estimators. At each given exogenous date $t$, we can associate a series of residuals normalized by the current value $y_t$. They are:
\begin{equation}
	\label{purenoncausal2}
\hat{U}_{t+h,t,T}=\hat{u}_{t+h,T}/y_t, \qquad h=-H,..., H,
\end{equation}
indexed by $h$. Then we have:
\begin{equation}
	\label{eq5.3}
	\hat{U}_{t+h,t,T}= \frac{\hat{u}_{t+h,T}-u_{t+h}}{y_t} +\frac{u_{t+h}}{y_t},
\end{equation}
where $$\frac{\hat{u}_{t+h,T}-u_{t+h}}{y_t}=-\Big(\hat{\phi}_{1,T}-\phi_1,..., \hat{\phi}_{p,T}-\phi_p \Big)\Big( \frac{y_{t+h-1}}{y_t}, ...,\frac{y_{t+h-p}}{y_t} \Big)'.$$
When $T$ is large,  and the future date $t$ is such that $y_t>y$, with $y$ large, we see that: 
\begin{align}
\frac{\hat{u}_{t+h,T}-u_{t+h}}{y_t} & \approx_{d} -\Big(\hat{\phi}_{1,T}-\phi_1,..., \hat{\phi}_{p,T}-\phi_p \Big) (X_{h-1},..., X_{h-p})', \\
\frac{u_{t+h}}{y_t}& \approx_{d} {U}_h,
\end{align}
where the approximation errors depend on the number of observations $T$ used for estimation, since $\Phi$ is replaced by $\hat{\Phi}_T$, as well as on the (large) value of $y$. Interestingly, if $h$ is such that $h \geq 1-N$, then ${U}_h=0$ by Corollary 1. Then it becomes deterministically known; in particular, this limit does not depend on the observed value of $y_{t+h}/y_t, h=-H,..., H$. More generally, under these asymptotic conditions and for standard estimation methods, such as approximate maximum likelihood, or generalized covariance approach,  we know that the asymptotically Gaussian variables $\sqrt{T}(\hat{\phi}_{1,T}-\phi_1,..., \hat{\phi}_{p,T}-\phi_p) \approx_d \kappa \sim N(0,\Omega_p)$, as well as the stochastic drift $N$ can be chosen independently of any finite number of  $y_{t+h}/y_t, h=-H,..., H$ around the exogenous future date $t$.  This can be used to derive the asymptotic distribution of the normalized residuals $\hat{U}_{t+h,t,T}, h=-H,..., H$. 

\begin{Proposition}
\label{band}
If $T$ tends to infinity and if at future date $t$, $y_t>y$, with $y$ tending to infinity sufficiently fast compared to $\sqrt{T}$, then 
conditional on $y_{t+h}/y_t, h=-H,..., H$,  we have:
$$
\sqrt{T} \Big[ \hat{U}_{t+h,t,T} \Big] \ \xrightarrow{d} -\Big[(\kappa_1,..., \kappa_p) (\frac{y_{t+h-1}}{y_t},..., \frac{y_{t+h-p}}{y_t})'\Big].
$$
with $h \geq 1-N$,  where $(\kappa_1,..., \kappa_p)'$ is a Gaussian vector with zero mean and a covariance matrix $\Omega_p$. 

If $(y_t)$ is MAR(1,1) Cauchy, and $y/\sqrt{T}$ tends to infinity, then the above weak convergence holds. 
\end{Proposition}
Note that we do not know \textit{a priori} what is a large value of $y$ and what is the realization of $N_t$ for date $t$. 
\begin{proof}
	See Appendix A.4.
\end{proof}
We can use this result to plot these sequences of normalized residuals (or transformations of such residuals),  possibly with their estimated confidence bands at 95\%, that are:
$$
\Big[ \hat{{U}}_{t+h,t,T} \pm \frac{1.96}{\sqrt{T}} \hat{\sigma}_{t+h,t,T}\Big]=\Big[ \frac{\hat{u}_{t+h}}{y_t} \pm \frac{1.96}{\sqrt{T}} \hat{\sigma}_{t+h,t,T}\Big],
$$
where $\hat{\sigma}_{t+h,t,T}^2=(\frac{y_{t+h-1}}{y_t},..., \frac{y_{t+h-p}}{y_t}) \hat{\Omega}_{p,T}(\frac{y_{t+h-1}}{y_t},..., \frac{y_{t+h-p}}{y_t})'$, and $\hat{\Omega}_{p,T}$ is a consistent estimator of ${\Omega}_{p}$.

\begin{Remark} The confidence intervals have been derived separately for each $h$, but joint confidence regions can be easily derived due to the asymptotic normality. 
\end{Remark}

At each date $t$, we get several confidence intervals $CI_{t+h, t,T}(u)$ and $CI_{t+h, t,T}(v)$, for the pure noncausal and pure (normalized) causal residuals.  Then at each exogenous date $t$, and maturity $h$, we can introduce indicator functions:
$$
\hat{I}_{t,h,T}(u)= \begin{cases}
1, & \text{ if the observations are such that } 0 \in CI_{t+h, t,T}(u)\\
0,& \text{ otherwise }
\end{cases},
$$
and a similar definition for $\hat{I}_{t,h,T}(v)$ for the pure causal normalized residuals.  Thus,  at each date $t$, we get an adjacency matrix of dimension $2 \times (2H+1)$ that summarizes the behaviour of these pure residuals.  

By Proposition \ref{band}, we expect the following pattern for the estimated adjacency matrices and their summaries.  If the MAR model is well specified and the date $t$ such that $y_t>y$, with $y$ large enough, then the estimated adjacency matrix is expected to have
a first row (resp. second row) with values 1 first (resp. 0 first), followed by 0 values (resp. 1 values) and coherent breaking maturities corresponding to the opposite $-N_t$ of the maturity of the peak.  Note also that, when such an estimated date $-\hat{N}_{t,T}$ appears, it can depend on $t$, under well-specified MAR process.  
Indeed, even if the distribution of $N_t$ given $y_t>y,$ $y$ large, does not depend on $t$, its realization $N_t$ can differ with the environment at date $t$. 

Similarly, we can also define the pure causal normalized residual by $\hat{V}_{t+h,t,t}/y_t$. 
Obviously, if a process $(y_t)$ is pure noncausal (resp. pure causal), then there is no need define its pure causal (resp. noncausal) normalized residual process. 
 
\subsection{Alternatives}
The results in Proposition 11 are valid if the MAR($p,q$) model is satisfied and at date $t$ with $y_t>y$ large. Let us now discuss what will arise under an alternative MAR($p^a, q^a$) with $p^a\geq p, q^a\geq q$.  Under this alternative,  we have:
$$
u_t^a= \Phi^a(L)y_t, \qquad v_t^a=\Psi^a(L^{-1}) y_t,
$$
and
$$
y_t=L^{q^a}b_1^a(L) v_t^a+b_2^a(L)u_t^a.
$$
We can still apply Proposition 1 and define the tail processes:
$$
(X^a_h),  \qquad (U_h^a)= (\Phi^a(\tilde{L})X^a_h),  \qquad (V_h^a)= (\Psi^a(\tilde{L}^{-1})X^a_h).
$$

Let us now consider what is arising when we apply Proposition 1 with the possibly mis-specified MAR($p, q$) model.  There are two effects:
\vspace{0.5em}

\noindent $i).$ The lag polynomials are mis-specified with the coefficients $\phi_j, \psi_j$ replaced by pseudo-true values $\phi^*_j, \psi^*_j$, say.  We denote by $\Phi^*(L)$,  $\Psi^*(L^{-1})$ these pseudo lag polynomials. 
\\

\noindent $ii). $ The pure tail process $(U^*_h)$,  $(V^*_h)$, computed as if the MAR($p,q$) model was satisfied, are such that:
$$
U^*_h=\Phi^*(\tilde{L})X^a_h, \qquad V^*_h=\Psi^*(\tilde{L}^{-1})X^a_h,
$$
under the alternative.  They differ from $(U_h)$,  $(V_h)$.  More precisely,  by eq. \eqref{pqvux},  we have:
\begin{align*}
U^*_h&=\Phi^*(\tilde{L})\Big[ \tilde{L}^{q^a} b_1^{a}(\tilde{L})V_h^a+b_2^{a}(\tilde{L}) U_h^a  \Big],\\
V^*_h&=\Psi^*(\tilde{L}^{-1})\Big[ \tilde{L}^{q^a} b_1^{a}(\tilde{L})V_h^a+b_2^{a}(\tilde{L}) U_h^a  \Big].
\end{align*}
In particular  these mis-specified pure tail processes will not take zero value for some $h$ at the difference of the well-specified $U_h^a$, $V_h^a$. 
\subsection{An illustration}
To facilitate the comparison with the applied literature on causal-noncausal process, we consider the same series of futures on corn used in Gouri\'eroux et al. (2021), estimated as a MAR(1,1) process with $\hat{\Phi}_T=-0.006$, $\hat{\Psi}_T=0.977$. \footnote{The dataset is downloaded from Professor Joann Jasiak's website: www.jjstats.com/papers/cornprice.txt} The series is plotted in Figure 3 and we will follow the conditional normalized causal and noncausal residual plots during the bubble episode.  The evolution of these residual plots is provided in Figure 2. 

\begin{figure}[H]
\centering
\includegraphics[scale=0.2]{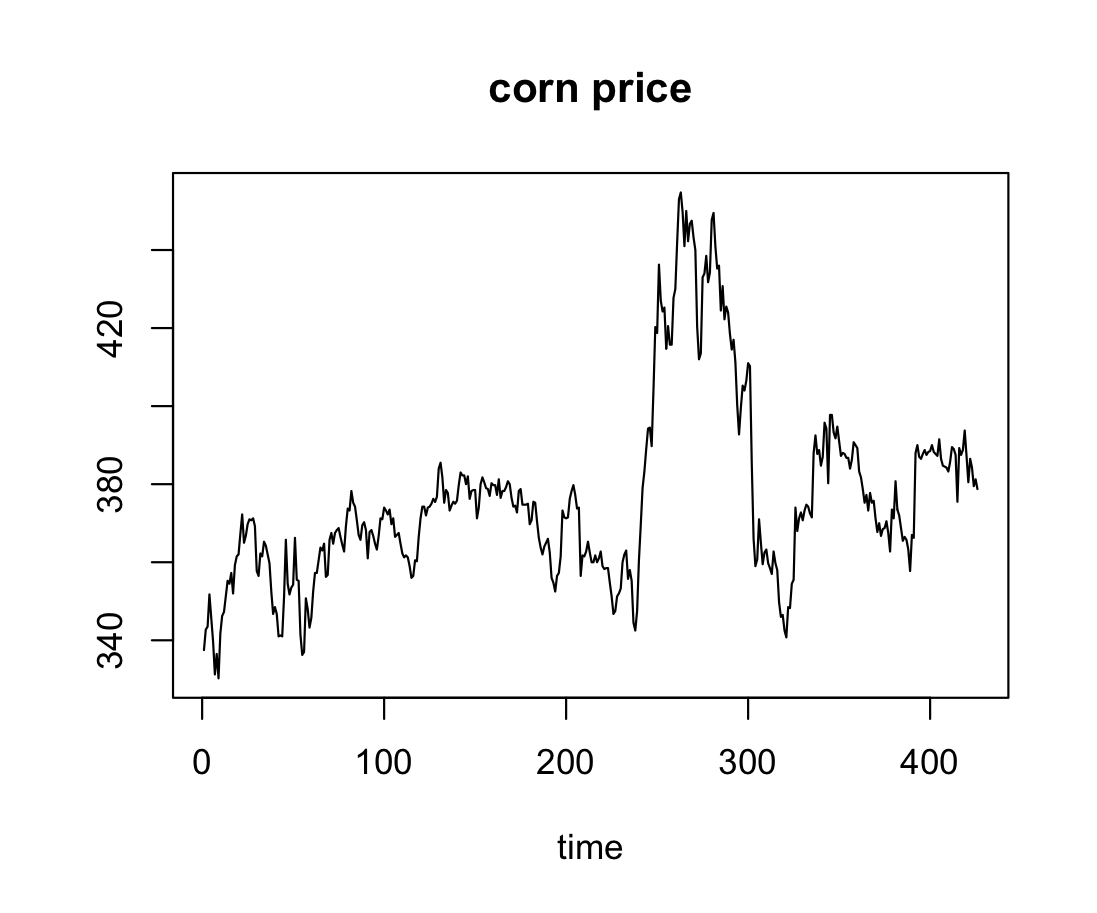}
\caption{Corn Futures Price between July 02, 2018,
and February 03, 2020. }
\end{figure}
In mid June 2019 (corresponding roughly to $t=260$ in Figure 3), corn prices experienced a significant surge, driven by a US Department of Agriculture report on June 26 about an unexpectedly slashed crop forecasts.  However, price fell sharply in August due to expectation of increased supply from the US and South America,  along with a faster-than-expected harvest in the US. 

\begin{figure}[H]
\centering
\includegraphics[scale=0.2]{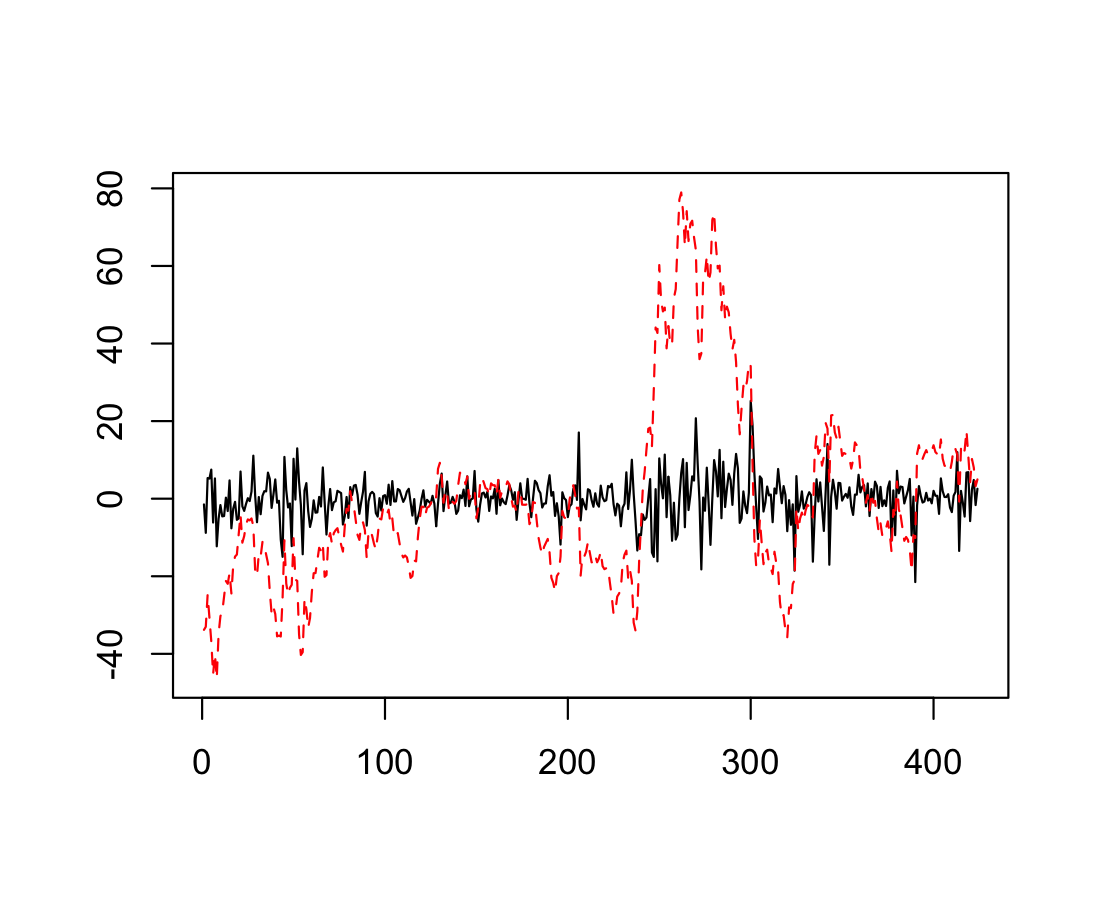}
\caption{Causal and Noncausal Residuals.  Red dashed line: noncausal residual $(\hat{u}_{t,T})$.  Black full line: causal residual $(\hat{v}_{t,T})$.}
\end{figure}
Let us now illustrate the normalized noncausal (resp. causal) residuals $\hat{U}_{t+h,t,T}$ (resp.$\hat{V}_{t+h,t,T}$) as functions of horizon $h$.  We choose three dates $t=210,  240,  350$,  where $t=240$ corresponds roughly to the peak of the corn future prices.  
Then we compute  $|\hat{U}_{t+h,t,T}|$ and $|\hat{V}_{t+h,t,T}|$ using eq. \eqref{purenoncausal2},  for $h=-30,...,+30$.  Figure  5  provide these plots for the three dates,  respectively. We take the absolute value in order to better visualize those values that are close to zero. The figures for $\hat{U}_{t+h,t,T}$ (resp. $\hat{V}_{t+h,t,T}$) are on a same scale to make more apparent the values close to zero.

\begin{figure}[H]
\centering
\includegraphics[width=12cm, height=10cm]{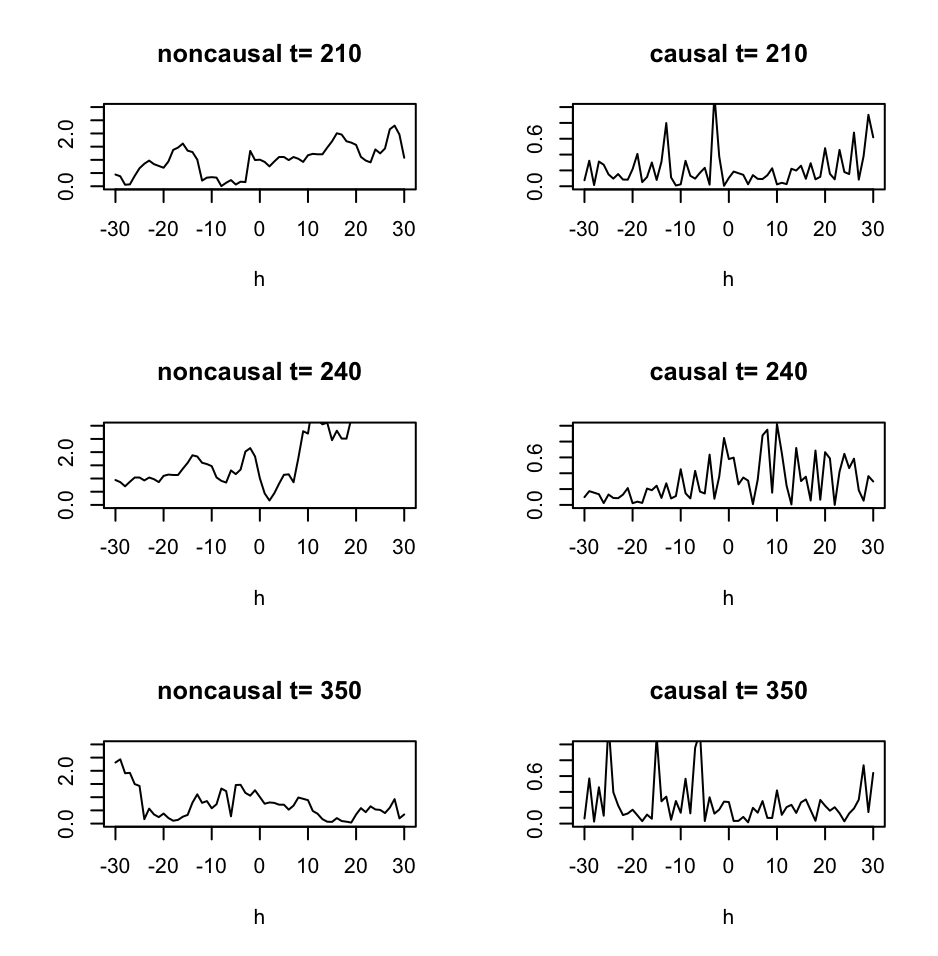}
\caption{Normalized Noncausal and Causal Residuals $|\hat{U}_{t+h,t}|$ and $|\hat{V}_{t+h,t}|$  with $h=-30,...,30$.   Upper panel: $t=210$;  Mid panel: $t=240$;  Lower panel: $t=350$.}
\end{figure}

According to Corollary 1,  $U_h$ is zero and $\hat{U}_{t+h,t}$ is also expected to be close to zero,  if $h+N_t \geq 1$,  that is when $h$ is large and positive.   This is the case in the noncausal plot in the lower panel ($t=350$,  corresponding to $N_t \approx 80$),  as well as the first half of the causal plot in the mid panel ($t=240$,  corresponding to $N_t \approx 0$).  The noncausal plot in the upper panel ($t=210$,  corresponding to $N_t \approx -30$),  on the other hand,  is not close to zero.  

Note that the information content of these normalized residual plots also depends on the goodness-of-fit of the model to the data at hand.  In the following,  we consider a limiting, hypothetical situation without misspecification,  nor estimation error (that is,  $\kappa_1$ in Proposition 1 is zero,  and the process of normalized residuals $(\hat{U}_{t+h, t, T})$ converges in distribution to process $(X_h)$).  We use the parameter values estimated by Gouri\'eroux et al. (2021) to simulate a MAR(1,1) Cauchy process,  and compute the residual plots on these simulated data.  We plot these residual plots for three dates,  respectively before, during and after the biggest realized bubble period.  Figures 6 and 7 are the counterparts of Figures 4 and 5, respectively.  

\begin{figure}[H]
\centering
\includegraphics[scale=0.2]{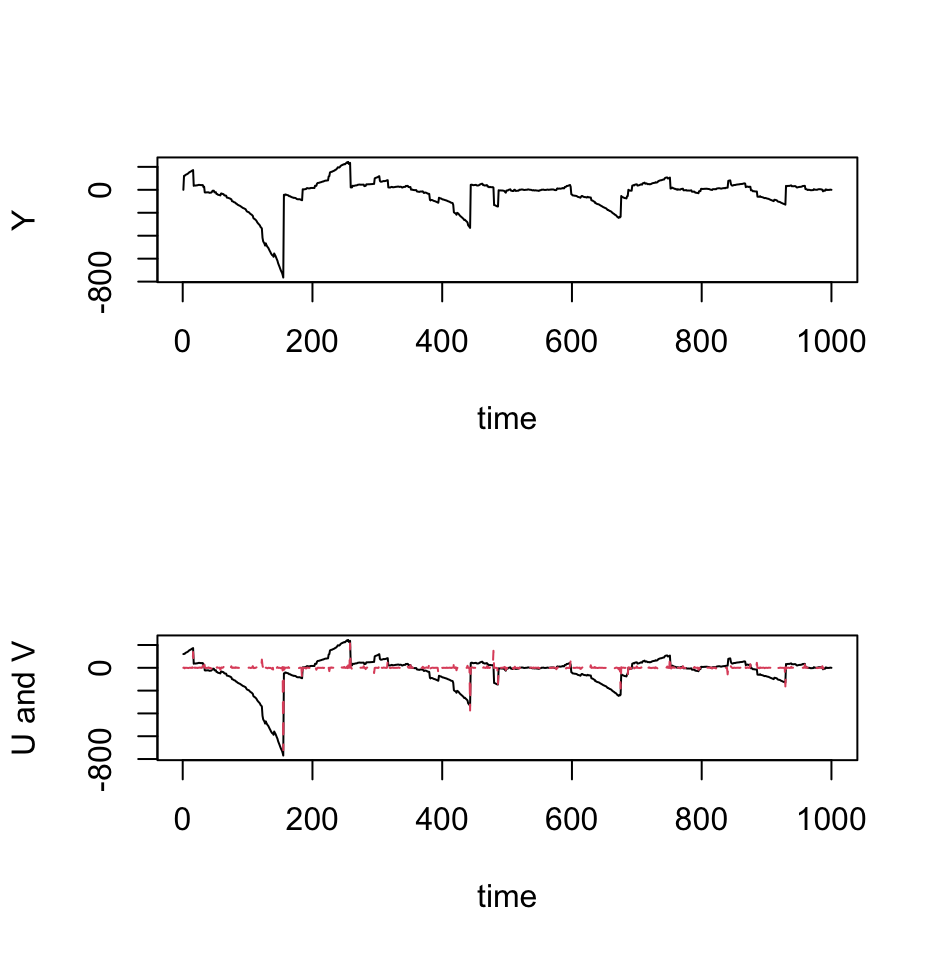}
\caption{Simulated trajectory of a MAR(1,1) Cauchy proces with ${\Phi}=-0.006$, ${\Psi}=0.977$ . }
\end{figure}

\begin{figure}[H]
\centering
\includegraphics[width=12cm, height=10cm]{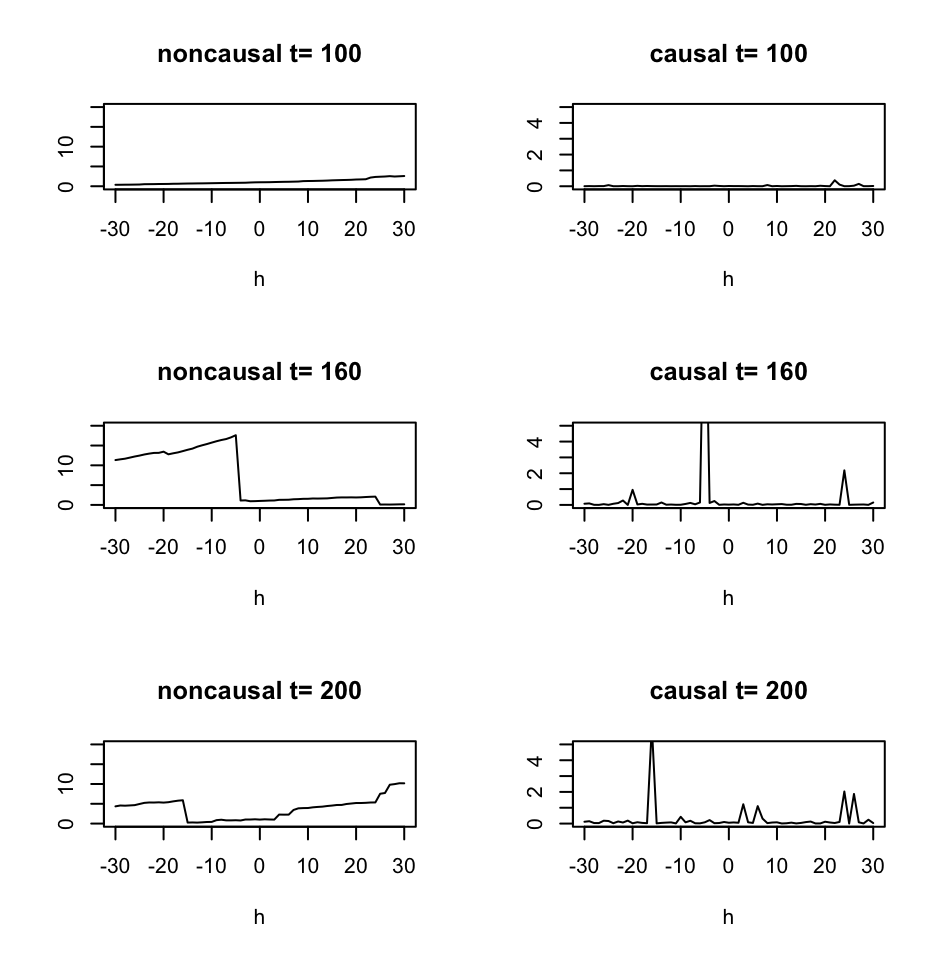}
\caption{Normalized Noncausal and Causal Residuals $|\hat{U}_{t+h,t}|$ and $|\hat{V}_{t+h,t}|$  with $h=-30,...,30$ computed on the simulated data.    Upper panel: $t=210$;  Mid panel: $t=240$;  Lower panel: $t=350$.}
\end{figure}

%


\section{Remarks and Further Developments}

All results of our paper have been derived and discussed for univariate causal-noncausal processes. However, there is an increasing literature on multivariate causal/noncausal processes both from the theoretical perspective [Lanne and Saikkonen (2013),  Gouri\'eroux and Jasiak (2016, 2024), Fries and Zako\"ian (2019), Davis and Song (2020), Fries (2022),  De Truchis et al. (2025)] and applied perspective [Cubbada et al. (2019, 2023)]. It is known that conditional extreme value theory is more difficult to develop in the multivariate framework. In our special framework the following questions will have to be solved:

i) How to account for the dimensions and dynamics of the pure causal and pure noncausal components?

ii) How to define conditioning sets that allow for deriving tail processes, while being interpretable? Do they have to be written by component of the process, on specific combinations (portfolios), or on some underlying factors?

iii) How to deal with the possibility of different tail indexes in the errors,  or the possibility of cointegrated bubbles?

\newpage

\appendix
\section{Appendix}
\subsection{Proof of Lemma 1}
 i) We have:

$
\begin{array}{lcl}
u_t = \Phi (L) y_t & = & \Phi (L) (\Sum^{+\infty}_{h=-\infty} c_h \epsilon_{t-h}) \\
                   & = & \Sum^p_{j=0} \varphi_j [ \Sum^{+\infty}_{h=-\infty} c_h \epsilon_{t-h-j}] \\
                   & = & \Sum^p_{j=0} \varphi_j ( \Sum^{+\infty}_{h=\infty} c_{k-j} \epsilon_{t-h}), \; \mbox{by a drift on the index},\\
                   & = & \Sum^{+\infty}_{k=-\infty} [(\Sum^p_{j=0} \varphi_j c_{k-j}) \epsilon_{t-k}] \\
                   & = & \Sum^{+\infty}_{h=-\infty} \{ [\Phi (\tilde{L}) c_h] \epsilon_{t-h}\}.
\end{array}
$

Therefore we deduce that: $ b_h = \Phi (\tilde{L}) c_h$.

ii) Similarly, we have:

$$
\begin{array}{lcl}
v_t = \Psi (L^{-1}) y_t & = & \Psi (L^{-1}) [\Sum^{+\infty}_{h=-\infty} c_h \epsilon_{t-h}]\\
                        & = & \Sum^q_{j=0} [\psi_j (\Sum^{+ \infty}_{k=-\infty} c_{k+j} \epsilon_{t-k})] \\
                        & = & \Sum^{+ \infty}_{k=-\infty} [ ( \Sum^\infty_{j=0} \psi_j c_{k+j}) \epsilon_{t-k}] \\
                        & = & \Sum^{+\infty}_{h=-\infty} \{ [\psi (\tilde{L}^{-1}) c_h] \epsilon_{t-h}\}.
\end{array}
$$

We deduce that: $a_h = \Psi (\tilde{L}^{-1}) c_h$.

iii) These equalities correspond to the truncations of the $(a_h)$ and $(b_h)$ moving-average series.

\subsection{Extremal Behaviour of Moving Average Processes when the Conditioning Event is a Large Exceedance}

\subsubsection{Extremal behaviour of the MAR(2,1)}

Let us assume that:

\begin{equation*}
(1-\lambda_1 L) (1-\lambda_2 L) (1-\mu L^{-1}) y_t = \epsilon_t,
\end{equation*}%

\noindent where $\lambda_1, \lambda_2, \mu $ are real of modulus smaller than one, or equivalently, %

\begin{equation*}
y_{t}=\sum_{h\in \mathbb{Z}}c_{h}\epsilon _{t-h},
\end{equation*}%

with

\begin{eqnarray*}
c_{h} &=&\frac{1}{(1-\lambda_1 \mu) (1-\lambda_2 \mu)}\mu^{-h}, h \leq 0, \\
c_{h} &=&\frac{1}{(1-\lambda_1 \mu) (1-\lambda_2 \mu)} b_h, h \geq 0, \\
\end{eqnarray*}%

and%

\begin{equation*}
b_{h}=\frac{\lambda^{h+1}_1 (1-\lambda_{2}\mu )-\lambda_{2}^{h+1}(1-\lambda_{1}\mu )}{%
\lambda_{1}-\lambda_{2}}.
\end{equation*}.

We have:

\begin{equation*}
\left\Vert c\right\Vert _{\alpha }^{\alpha }=\frac{1}{(1-\lambda_{1}\mu
)^{\alpha }(1-\lambda_{2}\mu)^{\alpha }}\left( \frac{1}{1-\mu^{\alpha }}%
+\gamma _{\alpha }\right),
\end{equation*}%

with%

\begin{equation*}
\gamma _{\alpha }=\sum_{h=1}^{\infty }b_{h}^{\alpha },
\end{equation*}%

and%

\begin{eqnarray*}
p_{j} &=&\left( \frac{1}{1-\mu ^{\alpha }}+\gamma _{\alpha }\right)
^{-1}\mu^{\alpha j},\quad j \leq 0, \\
p_{j} &=&\left( \frac{1}{1-\mu ^{\alpha }}+\gamma _{\alpha }\right)
^{-1}b_{j}^{\alpha },\quad j\geq 0. \\
\end{eqnarray*}

\subsubsection{Closed form sequence of moving average coefficients of the MAR($p,q$) process}

The aim of this subsection is to derive the closed form expression of the moving average coefficients $(c_h)$. This is a consequence of Lemma 1 iii), that is the fact that the $(c_h)$ satisfy backward/forward recursive equations. Then they can be derived from the roots, i.e. $\lambda_i, \mu_j$, their multiplicity orders and a set of initial/terminal conditions. More precisely, let us denote by $\lambda_i, i=1, \ldots, p^*, \mu_j, j=1,\ldots, q^*$ the distinct values of the $\lambda, \mu,$ respectively,  and $n_i, m_j$ their multiplicity orders, then the moving average coefficients have the form of exponential times polynomial functions of $h$:

\begin{eqnarray*}
  c_h & =& \sum^{p^*}_{i=1} [\lambda^h_i (\sum_{k=0}^{n_i+1} \alpha_{ik} h^k)], \mbox{for}\; h \geq p-1, \\
  c_h & =& \sum^{q^*}_{i=1} [\mu^{-h}_i \sum_{k=0}^{m_j+1} \beta_{jk} (-h)^k], \mbox{for}\; h \leq - q.
\end{eqnarray*}

These expressions are jointly valid for the indexes $h$ such that $1-p \leq h \leq -1+q$. By writing the equality of the two expressions for these values of $h$, we get a system of $p+q$ equations that can be solved to get the values of the $p+q$ parameters $\phi, \psi$.


\subsection{Extremal Behaviour of MAR  Processes when the Conditioning Event is a Large Value} 

We provide the limiting behaviour of the distribution of process $(y_t)$, conditional on a large value $y$ at date $t$. Let us consider the MAR$(p,q)$ process:

\begin{equation*}
\Phi \left( L\right) \Psi \left( L^{-1}\right) y_{t}=\epsilon _{t},
\end{equation*}
where:
\begin{eqnarray*}
\Phi \left( L\right) &=&1-\phi _{1}L-...-\phi _{p}L^{p}, \\
\Psi \left( L^{-1}\right) &=&1-\psi _{1}L^{-1}-...-\psi _{q}L^{-q},
\end{eqnarray*}%
where the roots of polynomial $\Phi $ and $\Psi $ are all outside the unit circle,  and their coefficients have alternating signs:
\begin{equation*}
\phi _{1}>0,\phi _{2}<0,\phi _{3}>0,...\quad \text{and}\quad \psi
_{1}>0,\psi _{2}<0,\psi _{3}>0,...
\end{equation*}%

These conditions ensure that the coefficients $c_h$ of the two-sided moving average representation of $y_t$ are nonnegative. We also assume that $\epsilon_t$ is almost surely positive.  

\vspace{1em}

Let us now prove that the conditional distribution of $(y_{t+h})_h/y_t$ converges to a discrete distribution with masses at $c_{n+h}/c_h$ as $y_t$ increases to infinity.  It suffices to show that for any integer $K$,  the finite dimensional process $(y_{t+h}, h=-K,..., K)/y_t$ converges.  

For $\varepsilon>0$,  we can find a positive integer $M$,  larger than $K$,  and such that $\displaystyle \frac{\sum_{h=-M}^M  c_h^{\alpha}}{\sum_{h=-\infty}^{\infty}  c_h^{\alpha}}>1-\alpha$.  In other words,   the total contribution of the $2M+1$ terms $\epsilon_{t+h}, h=-M,..., M$ to the tail of $y_t$ is at least $1-\varepsilon$.  Then we write $y_t$ into:
$$
y_t= \sum_{h=-M}^{M} c_h \epsilon_{t-h}+ \sum_{|h|>M}  c_h \epsilon_{t-h} . 
$$
The $2M+2$ terms on the right hand side are independent and have equivalent,  Paretian tails.  Thus we can apply Lemma $2'$ in Online Appendix B.3,  which says that, as $y_t$ increases to infinity,  the joint distribution of:
$(c_h\epsilon_{t-M},..., c_h\epsilon_{t+M}, \sum_{|h|>M}  c_h \epsilon_{t-h})/y_t$ converges to a multinomial distribution.  Thus the conditional distribution of $(y_{t+h}, h=-K,..., K)/y_t$ also converges to a discrete distribution with point masses at $c_{n+h}/c_n$.

\subsection{Asymptotic Distribution of the Residuals}
Proposition 11 assumes that the asymptotic error of the weak convergence in Proposition 1 for large $y$ is negligible compared to the asymptotic error due to the number $T$  of observations, which is known to be of order $1/\sqrt{T}$ under standard regularity conditions. The magnitude of the asymptotic error in $y$ in Proposition 1, however, is more complicated and requires more assumptions, such as in eq. \eqref{limiteqtail}, at which rate the ratio of pdf's $\frac{f_{\epsilon_t+\psi \epsilon_{t+1}(z) }}{f_{\epsilon}(z)}$ converges to its limiting value $1+\psi^\alpha$, as $z$ increases to infinity. This kind of property is called second-order regular variation in the extreme value literature and is well beyond the scope of the paper. However, from the proof of Lemma 2 given in Online Appendix B.1, we see that the weak convergence is based on the fact that for a given pdf $f$, the function $|s|f(s r)$ is the pdf of the random variable $\frac{R}{s}$, where $R$ is a random variable with pdf $f$. This latter random variable converges weakly to the point mass at zero, at the rate of $1/s$  as $s$ increases to infinity. As a consequence, asymptotic error of the weak convergence in Proposition 1 for large $y$ is at least of order $1/y$, if not slower. 

Note that this rate of ${1}/{y}$ can be attained in some cases. For instance, in the Cauchy MAR(1,1), we have explained in section 4.2.1 that the convergence rate is exactly ${1}/{y}$.  As a consequence, Proposition 11 holds, if $1/y$ goes to zero more quickly than $1/\sqrt{T}$. That is, in a double asymptotic framework where both $y$ and $T$ increase to infinity, the threshold $y$ should satisfy: $y/\sqrt{T} \rightarrow \infty$. 

\section{Online Appendix B:  Additional Proofs and Technical Details}
\subsection{Proof of Lemma 2}
\subsubsection{ Sketch of the proof}

Let us consider the change of variable $(S=Z_1+Z_2,R=\frac{Z_1}{Z_1+Z_2})$. The Jacobian is $|S|$, and the joint density of $(S, R)$ is:
\begin{equation}
	|s|f_1(sr)f_2((1-r)s),
\end{equation}
where $f_1$ and $f_2$ are the p.d.f.'s of $Z_1$ and $Z_2$, respectively. Thus the conditional distribution of $R$ given $S$ is:
\begin{equation}
	\label{conditionalrs}
	\ell(r|s)=|s|	\frac{f_1(sr)f_2((1-r)s)}{f(s)},
\end{equation}
where $f(\cdot)$ is the density of $Z_1+Z_2$.

Consider the conditional distribution \eqref{conditionalrs}. When $r$ is close to 0, we have: 
\begin{equation}
	\label{approximationrs}
\frac{f_2((1-r)s)}{f(s)} \approx \frac{s^{\alpha+1}}{(1+\xi) (1-r)^{\alpha+1}s^{\alpha+1}} = \frac{1}{(1+\xi) (1-r)^{\alpha+1}}. 
\end{equation} 
Because $(1-r)^{\alpha+1} \approx 1$ when $r$ is close to zero, eq. \eqref{conditionalrs} becomes:
$\ell(r|s) \approx \frac{1}{1+\xi}  |s|f_1(sr)$. 

Similarly, for $r \approx 1$, the conditional distribution becomes:
\begin{equation}
	\label{approximationrs2}
	\frac{f_1(rs)}{f(s)} \approx \frac{1}{1+\xi} \frac{1}{r^{\alpha+1}} ,
\end{equation}
hence $	\ell(r|s) \approx \frac{\xi}{1+\xi}  sf_2((1-r)s)$. 

Then we remark that $|s|f_1(sr)$ (resp. $|s|f_1(s(1-r))$ ) is the density of $\frac{R}{s}$ (resp. $\frac{1-R}{s}$). When $s$ goes to infinity, $\frac{R}{s}$ converges in probability to zero; it converges also weakly to the constant variable at 0. Thus the corresponding density converges weakly to the point mass at zero, in the sense that:
	$$
	\int s f_1(sr) g(r) \mathrm{d} r\rightarrow g(0),
	$$  
as  $s$ increases to infinity, 	for any integrable continuous,  bounded function $g$.  Here,  the domain of integration could be the entire real domain,  if $Z_1, Z_2$ are real valued,  or can be the domain of positive numbers,  if $Z_1$ and $Z_2$ are positively valued.  
	
\begin{Remark} The sequence of p.d.f.'s $|s|f_1(sr)$, indexed by $s$, is called Dirac sequence of measures in the literature, see e.g.  Kanwal (1998), section 3.3. 
\end{Remark}
\subsubsection{Formal Proof of Lemma 2} 
First, recall that the regular variation property $f(x)=\frac{l(x)}{x^{\alpha}}$ holds actually uniformly on any open set [see  Resnick (2008), Proposition 0.5] so long as it holds pointwise. Therefore, properties \eqref{approximationrs} (resp.  \eqref{approximationrs2}) also holds uniformly in $r$ for any $r$ such that $|r|<\delta$ (resp. $|r-1|<\delta$). 
when $s$ goes to infinity. 

Thus for any given $\epsilon>0$, we can choose a suitable $\delta>0$ ($\delta$ small) and a value $s_0>0$ ($s_0$ large) such that for any $s>s_0$,  
\begin{equation}
	\label{firstterm}
	|\frac{f_2(s(1-r))}{f(s)}-\frac{\xi}{1+\xi} |<\epsilon,
\end{equation}
for any $r $ such that $|r|<\delta$, and similarly,
\begin{equation}
	\label{secondterm}
	|\frac{f_1(rs)}{f(s)}- \frac{1}{1+\xi}|< \epsilon,
\end{equation}
 for any $r$ such that $|r-1|<\delta$. 

Let us now consider a function $g(\cdot)$ that is integrable, continuous, and bounded.  By continuity,  we can also assume, without generality, that $|g(r)-g(0)|<\epsilon$ for any $|r|<\delta$, and similarly $|g(r)-g(1)|<\epsilon$ for any $|r-1|<\delta$. Then we write the integral into three terms:
\begin{align*}
\int \ell(r|s)g(r) \mathrm{d}r&=\int_{|r|<\delta} \ell(r|s)g(r) \mathrm{d}r + \int_{|r-1|<\delta} \ell(r|s)g(r) \mathrm{d}r + \int_{|r| \geq \delta, |r-1| \geq \delta} \ell(r|s)g(r) \mathrm{d}r. 
\end{align*}
Let us evaluate separately the three terms. If $|r|<\delta$, then by eq. \eqref{approximationrs}, 
\begin{align}
|\int_{|r|<\delta} \ell(r|s)g(r) \mathrm{d}r -\frac{\xi}{1+\xi} g(0)| & \leq \int_{|r|<\delta} \ell(r|s)|g(r)-g(0)| \mathrm{d}r + |g(0)|  \Big| \int_{|r|<\delta} \ell(r|s) \mathrm{d}r -\frac{\xi}{1+\xi}\Big| \nonumber \\
&  \leq \epsilon \int_{|r|<\delta} \ell(r|s)  \mathrm{d}r + |g(0)|  \int_{|r|<\delta}  |s|f_1(sr)	\big|\frac{f_2((1-r)s)}{f(s)}-\frac{\xi}{1+\xi}\big|\mathrm{d}r \nonumber \\
& \qquad + |g(0)| \frac{\xi}{1+\xi} \int_{|r|\geq \delta}  |s|f_1(sr)	 \mathrm{d}r \nonumber \\
& \leq \epsilon+|g(0)|\epsilon+ |g(0)| \frac{\xi}{1+\xi} \int_{|z|\geq s\delta}   f_1(z)	 \mathrm{d}z, \nonumber
\end{align}
where in the last inequality we have used the change of variable $z=rs$. By tending $|s|$ to infinity, we have: $|\int_{|z|\geq s\delta}   f_1(z)	 \mathrm{d}z| \rightarrow 0 .$ Thus
\begin{equation}
		\label{around0}
\Big|\int_{|r|<\delta} \ell(r|s)g(r) \mathrm{d}r -\frac{\xi}{1+\xi} g(0) \Big|  \leq (2+|g(0)|)\epsilon, 
\end{equation}  for $|s|$ large enough. 

Similarly, we have: 
\begin{equation}
	\label{around1}
\Big|\int_{|r-1|<\delta} \ell(r|s)g(r) \mathrm{d}r -\frac{1}{1+\xi} g(1) \Big| \leq (2+|g(1)|)\epsilon,
\end{equation}
 for $|s|$ large enough. 
 
It suffices now to check that:
\begin{equation}
	\label{outside0and1}
	\int_{|r| \geq \delta, |r-1| \geq \delta} \ell(r|s)g(r) \mathrm{d}r \leq \epsilon,
\end{equation}
for $s$ large enough. This is due to the fact that outside 1 and 0, we have, uniformly in $r$: 
$\ell(r|s) \rightarrow 0$ when $s$ goes to infinity.  By the dominating convergence theorem, \eqref{outside0and1} is satisfied for $|s|$ large enough. 

As a consequence, by combining \eqref{around0}, \eqref{around1} and \eqref{outside0and1}, we have shown that:
$$
\int \ell(r|s)g(r) \mathrm{d}r \longrightarrow \Big[\frac{\xi}{1+\xi} g(0)+\frac{1}{1+\xi} g(1)\Big],
$$
as $|s|$ goes to infinity. 
\subsection{Proof of Lemma \ref{ultimatemonotone}}
First,  by Cline (1983, Theorem 2.3),  we have: 
$$
\lim_{z \to +\infty } \frac{S_{u}(z)}{S_{\epsilon}(z)} =
1+\psi^{\alpha}+\psi^{2\alpha}+\cdots =\frac{1}{1-\psi^{\alpha}},
$$
where $S_u$ (resp.  $S_{\epsilon}$) denotes the survival function of $u$ (resp. $\epsilon$).  

Then,  since $f_{\epsilon}$ is equivalent to an ultimately monotone function,  by the Monotone Density Theorem [Bingham, Goldie, Teugels (1989), Theorem 1.7.2,  page 39],  we have:
\begin{equation}
	\label{limiteqtail}
\lim_{z \to +\infty } \frac{S_{\epsilon}(z)}{f_{\epsilon}(z)} \frac{\alpha}{z} =1 ,
\end{equation}
 

Similarly,  we have:
$$
\lim_{z \to +\infty }  \frac{S_u(z)}{f_u(z)}\frac{\alpha}{z} =1.
$$

Combining the above three limits leads to:
$$
\lim_{z \to +\infty } \frac{f_{u}(z)}{f_{\epsilon}(z)} =\frac{1}{1-\psi^{\alpha}},  \text{ if } \psi \in (0,1).
$$
\subsection{Extension of Lemma 2}
Lemma 2 can be easily extended to include any finite number of independent variables. 
We have:

\textbf{Lemma $2'$: } If $Z_1,..., Z_n$ are independent,  with Paretian tails $f_i(z)=z^{-\alpha-1}l_i(z), i=1,..., n$,  where $l_1, ...,  l_n$ are slowly varying functions,  and if moreover they have equivalent p.d.f.'s: 
$$
\lim_{z \to \infty} \frac{f_1(z)}{f_i(z)} = \frac{\xi_1}{\xi_i}>0,  i=2,..., n,
$$
then the conditional distribution of: $$(\frac{Z_1}{Z_1+\cdots+Z_n}, \frac{Z_2}{Z_1+\cdots+Z_n}, \cdots, \frac{Z_n}{Z_1+\cdots+Z_n} ),$$ given $ S=Z_1+\cdots+Z_n=s$ converges weakly to the multinomial distribution with probabilities: $$\frac{\xi_1}{\xi_1+\cdots+\xi_n}, \frac{\xi_2}{\xi_1+\cdots+\xi_n}, ..., \frac{\xi_n}{\xi_1+\cdots+\xi_n}$$  respectively,  as $s$ increases to infinity. 

The proof of Lemma $2'$ has the same spirit as Lemma 2 and is omitted.  Lemma $2'$ leads immediately to Proposition \ref{h12}. 

\subsection{Proof of Proposition \ref{h12}}
We write:
\begin{align*}
y_t&=\frac{1}{1-\phi \psi} (v_t+ \psi u_{t+1})=\frac{1}{1-\phi \psi} (v_t+ \psi \epsilon_{t+1}+  \psi^2 u_{t+2}), \\
y_{t+1}&=\frac{1}{1-\phi \psi} (\phi v_{t}+\epsilon_{t+1}+ \psi u_{u+2}),\\
y_{t+2}&=\frac{1}{1-\phi \psi}(\phi^2 v_{t}+\phi \epsilon_{t+1} + u_{u+2}).
\end{align*}
Then we apply Lemma $2'$ to the three terms on the right hand side of the first equation,  which are independent and have equivalent Paretian tails,  and get Proposition \ref{h12}. 

\subsection{A discussion of the DBJ principle for noncausal processes with more than one noncausal roots}
Proposition 9 says that for any MAR$(p,1)$ process, the limiting distribution of $r_{t+1}$ given $y_t$ large, and known previous rates of increase $r_t, r_{t-1}$ behave in pretty much the same way as for a MAR$(1,1)$ process (see Proposition 8).  Indeed,  as explained at the beginning of section 4.2.2,  the analysis of the initial process $(y_t)$ is easily transformed into the analysis of its noncausal part $(u_t)$.

The aim of this subsection is to discuss the case where the number of noncausal roots $q$ is strictly larger than 1. To fix the ideas,  let us start with the case of MAR$(0,2)$ process,  with the infinite MA representation:
$y_t=\sum_{h=0}^{\infty} c_{-h} \epsilon_{t+h}$,  where $c_{-h}=b_h$,  whose expression is given by eq. \eqref{expressionbh}. In the following, we will show that,  contrary to a MAR(0,1), for which the limiting distribution of $r_{t+1}$ can only weight two values, in the MAR(0,2) case this limiting distribution is discrete but with an infinity of possible values. 
\subsubsection{First order approximation: the SBJ principle}
We start by analyzing the case where $r_t$ takes one of the limiting values predicted by Proposition 1.  By the SBJ principle,  since $y_t$ is large,   exactly one error among $\epsilon_t,\epsilon_{t+1},...,$ is large,  and we can distinguish two cases: 
\begin{enumerate}
\item if the SBJ is $\epsilon_t$,  then by Proposition 2,  the distribution of $(y_{t-1},y_{t+1})/y_t$ is approximately the point mass at $(b_1, 0)$.  In this case, we have:
\begin{equation}
\label{limit1}
r_{t+1} \approx 0, 
\end{equation}
and the bubble will collapse at the next period.  Thus,  as $y_t$ increases to infinity and $r_t$ tends to $b_1^{-1}$,  the conditional distribution of  $r_{t+1}$ given $ y_t, r_t$ converges to the point mass at zero.    
\item if the SBJ is $\epsilon_{t+h}$,  with a positive $h$,  then the distribution of $(y_{t-1},y_{t+1})/y_t$ is approximately the point mass at $(b_{h+1}/b_{h}, b_{h-1}/b_{h})$.  By the linear recursion between $b_{h-1}, b_{h}, b_{h+1}$ (see Proposition 3),  we get: $r_{t}^{-1} \approx \psi_1+\psi_2 r_{t+1}$,  or 
\begin{equation}
\label{limit2}
r_{t+1} \approx \frac{r_{t}^{-1}-\psi_1}{\psi_2},
\end{equation}
  which depends on the past value $r_t$,  as in Proposition \ref{yr}.  In other words,  if $r_t$ is close to any of the values $b_{h+1}/b_{h}, h=1,...$, the conditional distribution of $r_{t+1}$ given $ y_t, r_t$ converges to the point mass at $\frac{r_{t}^{-1}-\psi_1}{\psi_2}$. 
\end{enumerate}

\subsubsection{Second order approximation: the DBJ framework} Again,  in practice,  the realized value of $r_t^{-1}$ is almost surely different from $b_1$ and $b_{h+1}/b_{h}, h=1,...$. In other words, none of the above two limiting cases is accurate enough. We are back in a DBJ framework, where at least two errors are large and of comparable order. Indeed,  if among $\epsilon_t, \epsilon_{t+1},...$,  there were only one big jump, which we denote by $\epsilon_{t+i}$, then we would have approximately: 
$$y_t \approx b_i  \epsilon_{t+i} ,  \qquad y_{t-1} \approx b_{i+1}  \epsilon_{t+i},$$
which means that $r_t^{-1} \approx b_{i+1}/b_{i}.$

Moreover,  this DBJ framework is more complicated than in the MAR(1,1) case discussed in Section 4.2.1. Indeed, for a MAR(1,1) process, one of the big jumps concerns a past date and thus can be identified.  Hence the DBJ analysis of the initial process $(y_t)$ can be easily transformed into the SBJ of its noncausal part.   In the MAR(0,2) case, however, since each $y_t$ only depends on the future or current errors, both big jumps concern current or future dates and their exact dates are hence ``unknown" yet. More precisely,  we have to distinguish between two cases: 
\begin{enumerate}
\item Both big jumps concern future dates,  $t+i$ and $t+j$, with $0<i<j$.  Then we have:
\begin{align}
y_t & \approx  b_i \epsilon_{t+i}+ b_j \epsilon_{t+j}, \label{eq11}\\
y_{t-1} & \approx b_{i+1} \epsilon_{t+i}+ b_{j+1}\epsilon_{t+j} \label{eq12} ,\\
y_{t+1} &\approx  b_{i-1} \epsilon_{t+i}+ b_{j-1} \epsilon_{t+j} ,
\end{align}
with $b_{i-1}  \neq 0,$ since $i-1 \geq 0$.  By applying the recursive relationship between the coefficients: 
$$
b_{i+1}= \psi_1 b_i+ \psi_2 b_{i-1},  \qquad b_{j+1}= \psi_1 b_j+ \psi_2 b_{j-1},  
$$
we get:
$$
y_{t-1} \approx \psi_1 y_t+ \psi_2 y_{t+1},
$$
or 
$$
r_{t+1} \approx \frac{r_{t}^{-1}-\psi_1}{\psi_2}.
$$
That is,  eq. \eqref{limit2} also holds,  without $r_t^{-1}$ necessarily being of the form $b_{j+1}/b_j$.  Thus this first case is the extension of case 2 above under the SBJ framework.  
\item The two big jumps are $\epsilon_{t}$ and $\epsilon_{t+j}$, with $j>0$.  Then we have:
\begin{align}
y_t & \approx  b_0 \epsilon_{t}+ b_j \epsilon_{t+j}, \label{yt}\\
y_{t-1} & \approx b_{1} \epsilon_{t}+ b_{j+1}\epsilon_{t+j}, \label{ytm1}\\
y_{t+1} &\approx  b_{j-1} \epsilon_{t+j},\label{ytp1}
\end{align}
where $\epsilon_t$ no longer enters into the last equation for $y_{t+1}$.  This is the extension of case 1 under the SBJ framework. Then solving for $\epsilon_{t+j}$ by using the first two equations and plugging into the last equation leads to:
\begin{equation}
\label{hrv}
r_{t+1} =\frac{y_{t+1}}{y_t}\approx \frac{r_t^{-1}-b_1}{b_{j+1}-b_1b_j} b_{j-1}.
\end{equation}
Moreover,  as $j$ increases to infinity,  we have: $b_j/b_{j-1} \rightarrow \mu_1$.  Thus the right hand side of eq. \eqref{hrv} converges to the right hand side of eq. \eqref{limit2}. 
\end{enumerate}

 

Thus the limiting distribution of $r_{t+1}$ given $r_t$ is a countable discrete mixture,  with one point mass given by eq. \eqref{limit2},   as well as an infinity of point masses given by eq. \eqref{hrv} $j=1,2,...$ varying.  Note that whether or not we are in case 1 and 2 is uncertain even given $y_t$ and $r_t$,  and in case we are in case 2,  the exact value of the index $j$ is also uncertain.  

This finding is to be compared with the SBJ analysis obtained in eqs. \eqref{limit1} and \eqref{limit2},  where the limiting distribution of $r_{t+1}$ is essentially deterministic [see also De Truchis et al. (2025) for similar findings].  However,  it is easily checked that: $i)$ if $r_t=b_1^{-1}$,  then the right hand side of eq. \eqref{hrv} becomes 0,  and we recover case 1 above; $ii)$ if $r_t= \frac{b_h}{b_{h+1}}$,  then the right hand side of eq. \eqref{hrv} becomes $\frac{b_{h-1}}{b_h}$.  Therefore,  for these special values,  the DBJ reduces to the standard SBJ.  This is expected,  since the DBJ is a kind of refinement of the SBJ.

While the limiting distribution of the MAR(0, 2) is significantly more complicated than the limiting distribution of the MAR(0,1),  it also provides more flexibility than the latter for the trajectory of a noncausal process during a bubble,  since a MAR(0, 1) can only allow the same rate of increase $\psi^{-1}$ of the bubble.  


The analysis and results remain essentially the same for the MAR($0, q$) processes with $q \geq 3$ and with conditioning set $y_t=y$ large, and given the previous $q-1$ ratios $r_t, r_{t-1},...,r_{t+2-q}$.  Instead of considering DBJ for $q=2$,  we would need triple big jump (TBJ) for $q=3$, and so on.  

Finally,  as in Proposition 9,  the derivation of the limiting distribution of a MAR($p, q$) process $(y_t)$ can be obtained from the limiting distribution of the MAR($0, q$) process ($u_t$),  through a known  change of variable.  

\subsubsection{An illustration}
As an illustration, we consider a MAR(0,2) model in which $(\epsilon_t)$ is Cauchy $\mathcal{C}(0,1)$, and the two noncausal roots are $\mu_1=0.4$ and $\mu_2=0.6$. The predictive distribution does not have a tractable expression, thus we resort to Monte Carlo simulations.  We simulate a trajectory from this model, with $T=1000000$.  We first plot the conditional distribution of $r_{t+1}$ given $y_t>50$,  where the threshold 50 corresponds roughly to the 97.5 \% quantile of the marginal distribution of the process.  
\begin{figure}[H]
\centering
\includegraphics[scale=0.2]{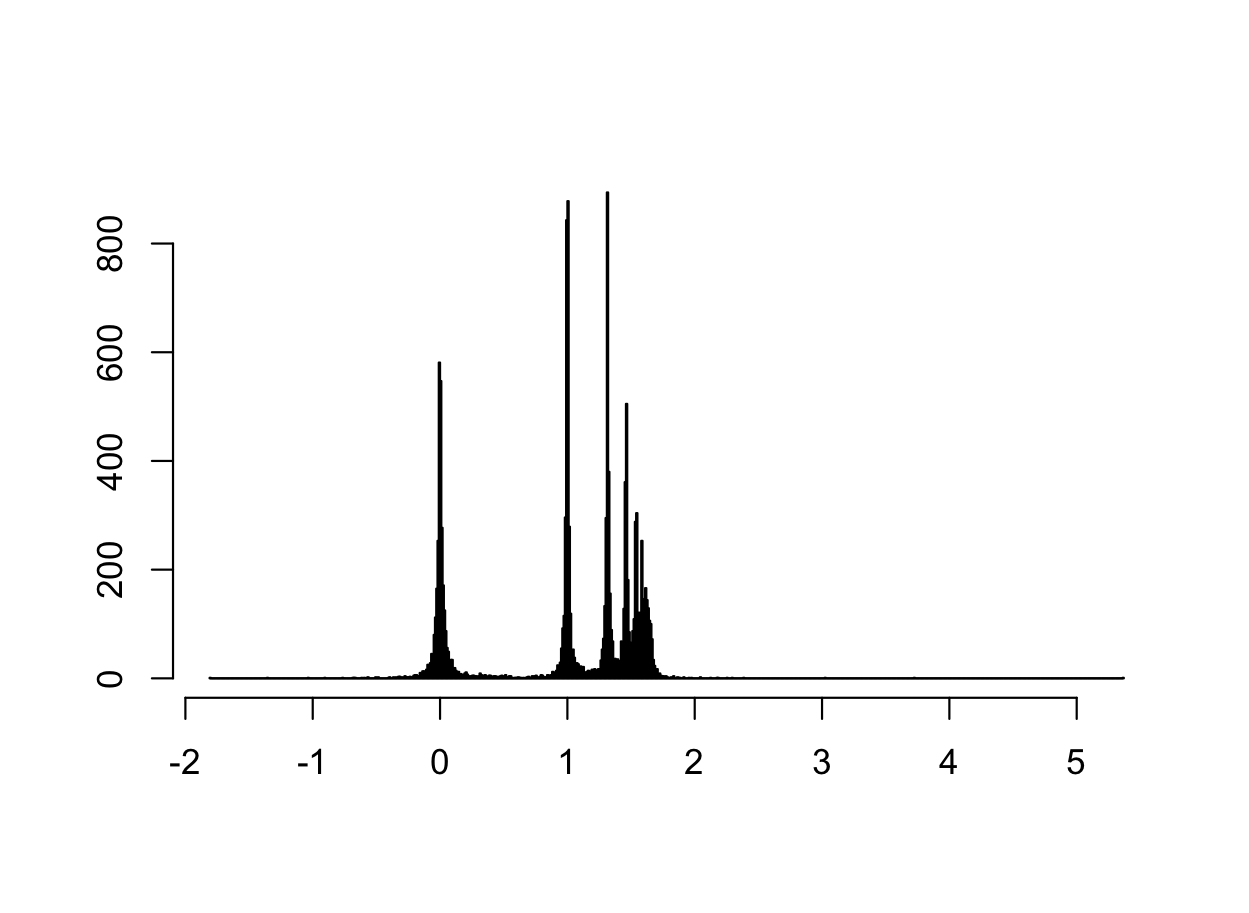}
\caption{Empirical Histogram of $r_{t+1}$ given $y_t>50$.}
\end{figure}
Figure 8 illustrates Proposition 1 in the case of MAR(0,2).  According to section 3.2.5,  this conditional distribution converges to a discrete distribution weighting the following values: 0 (corresponding to collapse of the bubble), $\mu_1^{-1}=1.66,  (\mu_1+\mu_2)^{-1}=1$,  as well as $\frac{\mu_1-\mu_2}{\mu_1^h-\mu_2^h}$, $h>0$ varying.  In particular,  as $h$ increases to infinity,  this latter ratio converges to $(\mu_1+\mu_2)^{-1}=1$.     These features are quite well reflected by the empirical histogram,  which has large peaks at 0,  1,  1.6,  but also several smaller peaks between 1 and 1.6. 

Let us now illustrate the SBJ of Section B.5.1.  Figure 9 plots the empirical histogram of $r_{t+1}$, given $y_t>50$,  and $1<r_t<1.1$,  that is close to $b_1^{-1}$.  By eq. \eqref{limit1}, we expect this conditional distribution to be concentrated around zero.  This is consistent with the histogram below.  
\begin{figure}[H]
\centering
\includegraphics[scale=0.2]{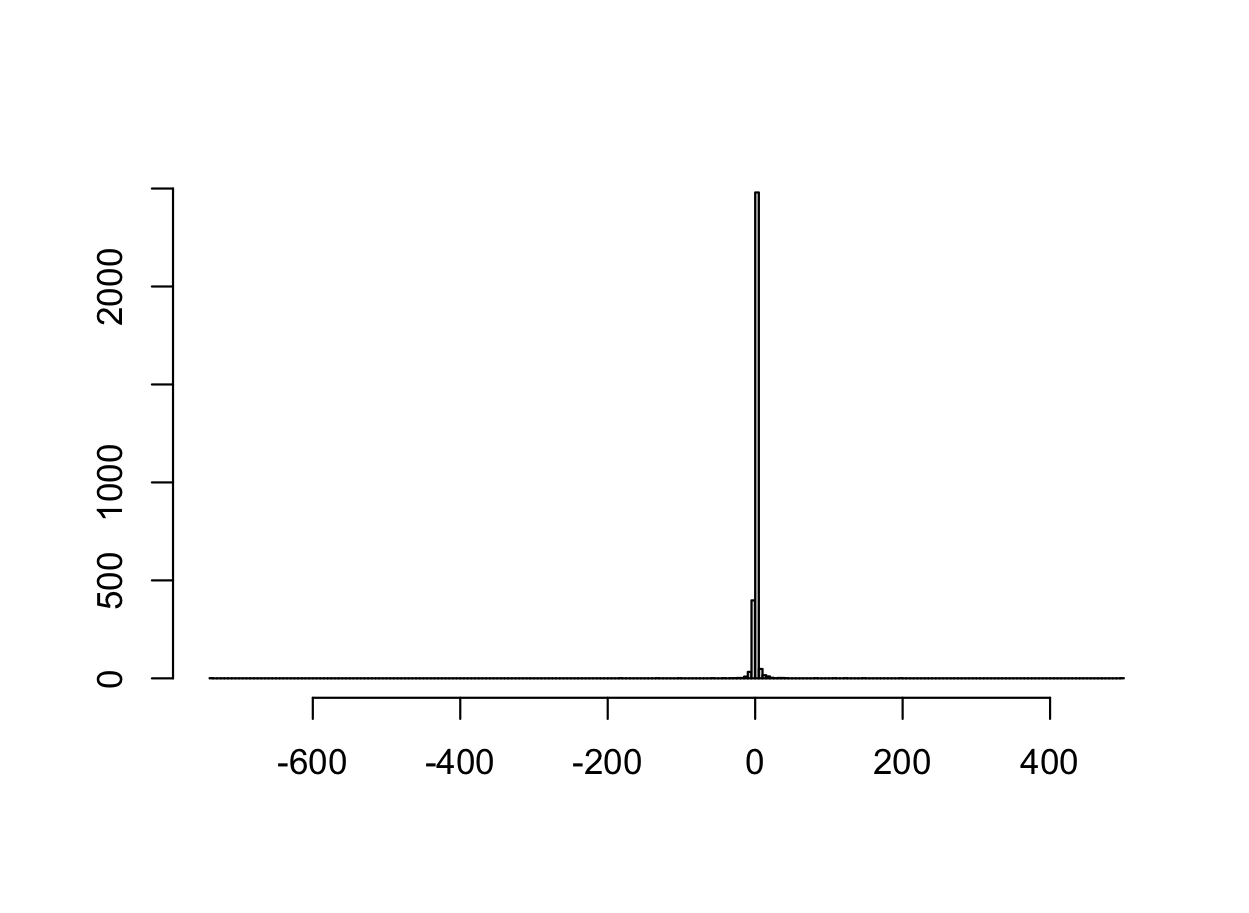}
\caption{Empirical Histogram of $r_{t+1}$ given $y_t>50$ and $1<r_t<1.1$.  Note that due to the extra conditioning event $1<r_t<1.1$,  the number of observations at our disposal to plot this histogram is significantly smaller than for the previous histogram. }
\end{figure}

Let us finally illustrate the DBJ of Section B.5.2.  Figure  10 is similar to Figure 9,  except that we replace $1<r_t<1.1$ by $2<r_t<2.1$.  Thus the value of $r_t$ is chosen to be different from $b_1$ and $b_{h+1}/b_h$, $h$ varying.  Section B.5.2. predicts that this conditional distribution is likely concentrated around $\frac{r_{t}^{-1}-\psi_1}{\psi_2} \approx 2$,  as well as $\frac{r_t^{-1}-b_1}{b_{j+1}-b_1b_j} b_{j-1}$, $j$ varying, with the latter sequence converging to 2 as $j$ increases.  This is consistent with Figure 10.  

\begin{figure}[H]
\centering
\includegraphics[scale=0.2]{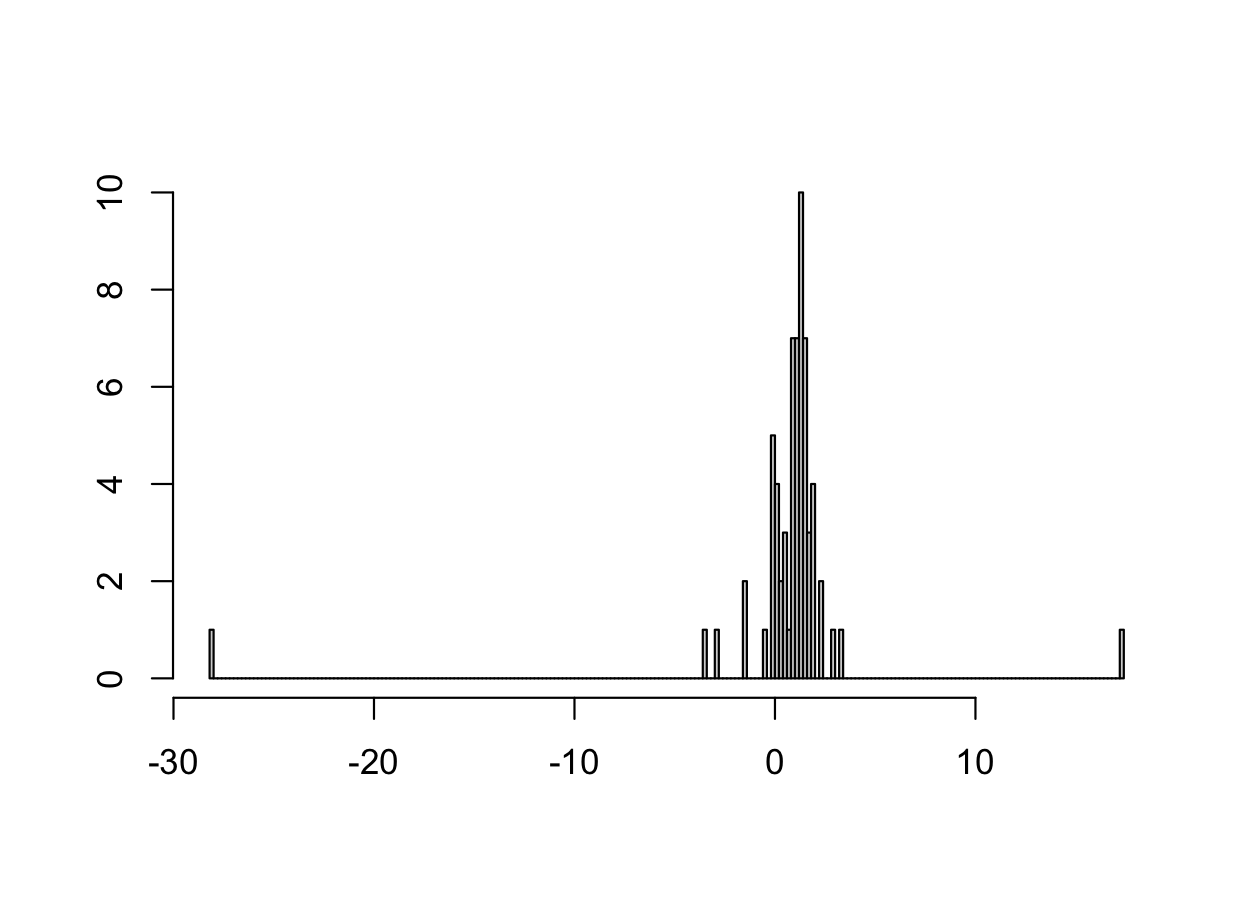}
\caption{Empirical Histogram of $r_{t+1}$ given $y_t>50$ and $2<r_t<2.1$. }
\end{figure}
By comparing these three histograms,  we see clearly the huge impact of the value of $r_t$ on the limiting distribution of $r_{t+1}$, even after conditioning upon $y_t$ large.

\end{document}